\documentclass[
aps,
prl,
10pt,
amssymb,
amsmath,
superscriptaddress,
tightenlines,
twocolumn,
notitlepage,
]{revtex4-2}

\usepackage{graphicx}
\usepackage{amsthm}
\usepackage{dcolumn}
\usepackage{bm}
\usepackage{xcolor}
\usepackage{comment}
\usepackage{bbold}
\newtheorem{lemma}{Lemma}
\newtheorem{theorem}{Theorem}

\newtheorem{proposition}{Proposition}
\newtheorem{remark}{Remark}
\begin{document}

\title{Exact Magnetic Confinement of a Doped Hole by Kinetic Frustration}

\author{Nikolai Peshcherenko}
\altaffiliation{These authors contributed equally to this work.}
\affiliation{Max Planck Institute for Chemical Physics of Solids, 01187, Dresden, Germany}

\author{Yixin Zhang}
\altaffiliation{These authors contributed equally to this work.}
\affiliation{Department of Physics and Astronomy, University of Tennessee, Knoxville, TN 37996, USA}
\affiliation{School of Physics, Peking University, Beijing 100871, China}

\author{Cristian D. Batista}
\affiliation{Department of Physics and Astronomy, University of Tennessee, Knoxville, TN 37996, USA}

\author{Yang Zhang}
\affiliation{Department of Physics and Astronomy, University of Tennessee, Knoxville, TN 37996, USA}
\affiliation{Min H. Kao Department of Electrical Engineering and Computer Science, University of Tennessee, Knoxville, Tennessee 37996, USA}

\begin{abstract}
We establish an exact mechanism of magnetic confinement driven purely by kinetic frustration. In a finite open sawtooth $U=\infty$ Hubbard chain of $L$ unit cells doped with one hole, the ground state at magnetization $M$ contains $M$ polarized triplet dimers surrounding a connected singlet region of length $\ell_{\rm s}=L-M$. A magnetic field selects $M$ and hence the confinement length. The hole forms a standing wave in the singlet region with exponential tails in the triplet background. Confinement requires neither disorder nor an exchange-generated string tension: destructive interference between direct and apex-mediated hopping makes the triplet bonds reflecting. At $t_2/t_1=[2+4\cos(\pi/(\ell_{\rm s}+2))]^{-1/2}$, they become perfectly reflecting, producing a position-degenerate compact localized state.
At every fixed field $0<h<h_{\rm sat}$, the singlet region remains finite as $L\to\infty$ and shrinks stepwise with increasing field, whereas at $h=0$ it spans the system.
\end{abstract}

\maketitle
\textit{Introduction.---}
Since Anderson's seminal work~\cite{Anderson1958}, disorder has provided the canonical route by which extended electronic wavefunctions break down~\cite{Lee1985,Evers2008,abrahams1979,50years}. Several mechanisms can suppress transport without quenched disorder. Stark localization relies on an imposed spatial gradient~\cite{Schulz2019,vanNieuwenburg2019,Scherg2021}, while Hilbert-space fragmentation~\cite{Sala2020,Khemani2020,Moudgalya2022,Rakovszky2020}, lattice-gauge constructions with emergent disorder~\cite{Smith2017,Brenes2018}, and interaction-generated energy mismatches~\cite{Schiulaz2015,DeRoeck2014,Yao2016} primarily constrain many-body dynamics in configurational space~\cite{Basko2006,Gornyi2005,Nandkishore2015,Abanin2019}. Real-space confinement can instead arise from destructive interference, most prominently through compact localized states in flat-band lattices~\cite{flat_band_1,Leykam2018,Danieli2021}, where the quantum cage is fixed by the single-particle geometry. Here we establish a different possibility: interactions organize the surrounding spins into reflecting boundaries whose separation is controlled by a uniform magnetic field. The field produces no spatial potential; it selects a spin sector whose texture cages the charge.

Kinetically frustrated doped Mott insulators naturally provide this mechanism. A hole permutes the spins along its trajectory~\cite{Nagaoka_FM,tasaki1989extension,tasaki1998nagaoka}, so alternative paths around a frustrated loop interfere constructively or destructively according to the spin state they traverse~\cite{shastry_kinetic,Pollmann2008,Sposetti2014}. On corner-sharing lattices, this constraint can select exact ground states: chains and trees of triangles stabilize a hole in a valence-bond-solid background of singlet dimers~\cite{kim2023exact,sawtooth_zhang}, while the dilute pyrochlore realizes a resonating-valence-bond liquid bound by hole kinetics rather than superexchange~\cite{kinetic_frustration_pyrochlore}. For the uniformly coupled sawtooth chain, Zhang and Glittum~\cite{sawtooth_zhang} obtained the thermodynamic-limit energy $E_\infty=-(1+\sqrt{5})t$ and a closed-form wavefunction for the unpolarized ground state, in which the singlet background spans the system and the hole remains extended. This raises two questions. How does the interference change when the ratio $t_2/t_1$ between direct basal and apex-mediated hopping is varied? More importantly, can a magnetic field organize polarized dimers into boundaries that confine the hole to a region that remains finite in the thermodynamic limit?

We answer both by solving exactly the two-parameter $(t_1,t_2)$ sawtooth $U=\infty$ Hubbard chain with one hole in every magnetization sector. For a finite open chain of $L$ unit cells, a unique dimer-covering mapping~\cite{kim2023exact} reduces the many-body problem to a single-particle tight-binding chain whose bond amplitudes depend on the singlet or triplet character of each dimer. We prove that the ground state at magnetization $M$ contains exactly $M$ polarized triplet dimers and one connected singlet region of length $\ell_{\rm s}=L-M$, in which the hole forms a standing wave with exponential tails of closed-form penetration depth. A Zeeman field $h$ selects $M$ and thereby the size of this interaction-generated quantum cage. The cage size remains finite as $L\to\infty$ at every fixed field $0<h<h_{\rm sat}$ and shrinks in discrete steps with increasing field. The cage walls arise because direct and apex-mediated hopping interfere destructively across triplet bonds; at $t_2/t_1=t_{\ell_{\rm s}}$ [Eq.~\eqref{eq:tells_main}] they become perfectly reflecting and the hole forms an exact compact localized state whose degeneracy grows linearly with system size.

\textit{Model.---}
We consider the one-hole-doped, infinite-$U$ Hubbard model on the
sawtooth lattice~\cite{Hubbard1963,Imada1998,LiebWu1968,
Sen1996,Blundell2003,Heinze2021}
[Fig.~\ref{fig:saw_tooth}(a)]:
\begin{equation}
H=-t_1\sum_{\langle i,j\rangle,\sigma}
c^\dagger_{i\sigma}c_{j\sigma}
-t_2\sum_{\langle\!\langle i,j\rangle\!\rangle,\sigma}
c^\dagger_{i\sigma}c_{j\sigma}
-hS^z_{\rm tot}.
\label{eq:hubbard_ham}
\end{equation}
Here $\langle i,j\rangle$ and
$\langle\!\langle i,j\rangle\!\rangle$ denote, respectively,
the sides and bases of the triangles, with each bond included in
both orientations, $i\to j$ and $j\to i$. The hopping amplitudes
satisfy $t_1,t_2>0$, the $U=\infty$ projection enforces no
double occupation, and $h=g\mu_{\rm B}B$ is Zeeman field in energy units.
Because the hopping commutes with $S^z_{\rm tot}$,
the field only shifts the energies of each magnetization sector,
$E_M(h)=E_M-hM$ with $M\equiv S^z_{\rm tot}$, so the
field-dependent ground state follows by minimizing $E_M-hM$
over $M$. By spin-inversion symmetry, it suffices to
consider $M\geq0$.

The open chain contains $L$ corner-sharing triangles and $2L+1$
sites; one hole leaves $2L$ electrons and $M_{\rm sat}=L$.
For every hole position, the remaining sites admit a
unique perfect matching into $L$ dimers
[Fig.~\ref{fig:saw_tooth}(b)], each in a singlet ($S=0$) or
triplet ($S=1$) state. In the basis $|i,\{S\},\{S^z\}\rangle$,
with $i$ the hole position, decomposing the hopping over the
corner-sharing triangles gives the exact single-particle model
\begin{equation}
H_{\rm TB}
=-t_1\sum_{\langle i,j\rangle}b_i^\dagger b_j
-t_2\sum_{\langle\!\langle i,j\rangle\!\rangle}
\eta_{ij}b_i^\dagger b_j-hM,
\label{eq:H_TB}
\end{equation}
where $\eta_{ij}=+1$ for a singlet dimer and $-1$ for a triplet
dimer. As in Eq.~\eqref{eq:hubbard_ham}, both bond sums include the two orientations. The hopping depends on the dimer spin only through $S$,
not $S^z$. Since each triplet contributes at most one unit to
$S^z_{\rm tot}$, the number $k$ of triplet dimers satisfies
$k\geq M$, and the fixed-$M$ ground state follows by minimizing
the lowest eigenvalue of Eq.~\eqref{eq:H_TB} over $k\geq M$ and
over the placement of the triplets.

\begin{figure}[t]
\centering
\includegraphics[width=0.75\linewidth]{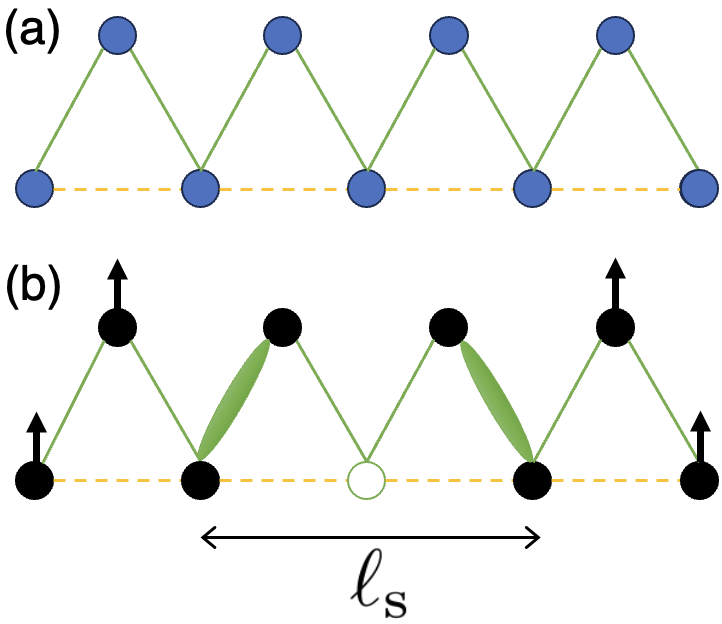}
\caption{
(a) Sawtooth lattice with hopping amplitudes $-t_1$ on solid
bonds and $-t_2$ on dashed bonds.
(b) Unique dimer covering for a fixed hole position (open circle).
Singlet dimers are enclosed by ovals, triplet dimers are
indicated by parallel spins. In the resulting tight-binding
model, the basal hopping is $-t_2$ across
a singlet and $+t_2$ across a triplet dimer. The singlet region of length $\ell_{\rm s}$ is indicated.}
\label{fig:saw_tooth}
\end{figure}

\textit{Sector ordering.---}
We first prove that the ground-state energy $E_{\rm GS}(k)$ at fixed $k$ (and $h=0$) strictly increases with $k$. The key step is an exact mapping of Eq.~\eqref{eq:H_TB} to a strictly one-dimensional problem [Fig.~\ref{fig:proof}(a)].
\begin{lemma}[1D mapping]
    \label{lemm:1D_mapping}
    For any fixed dimer configuration $\mathcal{C}$, eliminating the apex amplitudes via $Ea_{2n}=-t_1(a_{2n-1}+a_{2n+1})$ yields an energy-dependent matrix $\mathcal{K}_\mathcal{C}(E)$ on the $L+1$ basal sites, with end-site potentials $U_{\rm b}=t_1^2/|E|$ and effective singlet and triplet hoppings
    \begin{align}
        \tilde t_s=\frac{t_1^2}{|E|}+t_2,\qquad\tilde t_t=\left|\frac{t_1^2}{|E|}-t_2\right|,\label{eq:eff_hoppings}
    \end{align}
    satisfying $\tilde t_s>\tilde t_t\geq0$. The ground-state energy $E_\mathcal{C}$ of Eq.~\eqref{eq:H_TB} is the unique negative solution of $\lambda_{\min}[\mathcal{K}_\mathcal{C}(E)]=E$.
\end{lemma}
\begin{proof}
Since the basal graph is a tree, all hopping signs can be gauged away. Sylvester's law of inertia applied to the Schur complement identifies $E_{\mathcal C}$ with the unique negative solution; see Sec.~\ref{app:1D_mapping} of the Supplemental Material (SM)~\cite{SM}.
\end{proof}

\begin{theorem}[Sector ordering]
\label{thm:sector}
$E_{\rm GS}(k_1)< E_{\rm GS}(k_2)$ whenever $k_1<k_2$.
Consequently, at $h=0$ the global ground state lies
in the all-singlet sector, $k=M=0$.
\end{theorem}

\begin{proof}
It suffices to compare two dimer configurations of the mapped 1D problem from Lemma~\ref{lemm:1D_mapping} that differ by a single bond. Let $\mathcal{C}$ contain a triplet on the bond $(i,i+1)$ and let $\mathcal{C}'$ be obtained by replacing it by a singlet. With $\phi\ge0$ the ground state of $\mathcal{K}_\mathcal{C}(E)$, the variational principle gives
\begin{align}
&\lambda_{\min}[\mathcal K_{\mathcal C'}(E)]
\leq
\langle\boldsymbol{\phi}|
\mathcal K_{\mathcal C'}(E)
|\boldsymbol{\phi}\rangle=
\nonumber\\
&=
\lambda_{\min}[\mathcal K_{\mathcal C}(E)]
-2\bigl(\tilde t_s-\tilde t_t\bigr)
\phi_i\phi_{i+1}\le
\lambda_{\min}[\mathcal K_{\mathcal C}(E)].
\label{eq:variational_order}
\end{align}
For $\tilde t_t>0$, the Perron--Frobenius theorem gives $\phi_i\phi_{i+1}>0$, so the inequality is strict; the case $\tilde t_t=0$ is treated in Sec.~\ref{app:ordering} of the SM~\cite{SM}. Evaluating the strict inequality at $E=E_\mathcal{C}$ and applying Lemma~\ref{lemm:1D_mapping} gives $E_{\mathcal{C}'}<E_\mathcal{C}$; iterating proves the theorem.
\end{proof}

\begin{figure}[t]
\centering
\includegraphics[width=0.95\linewidth]{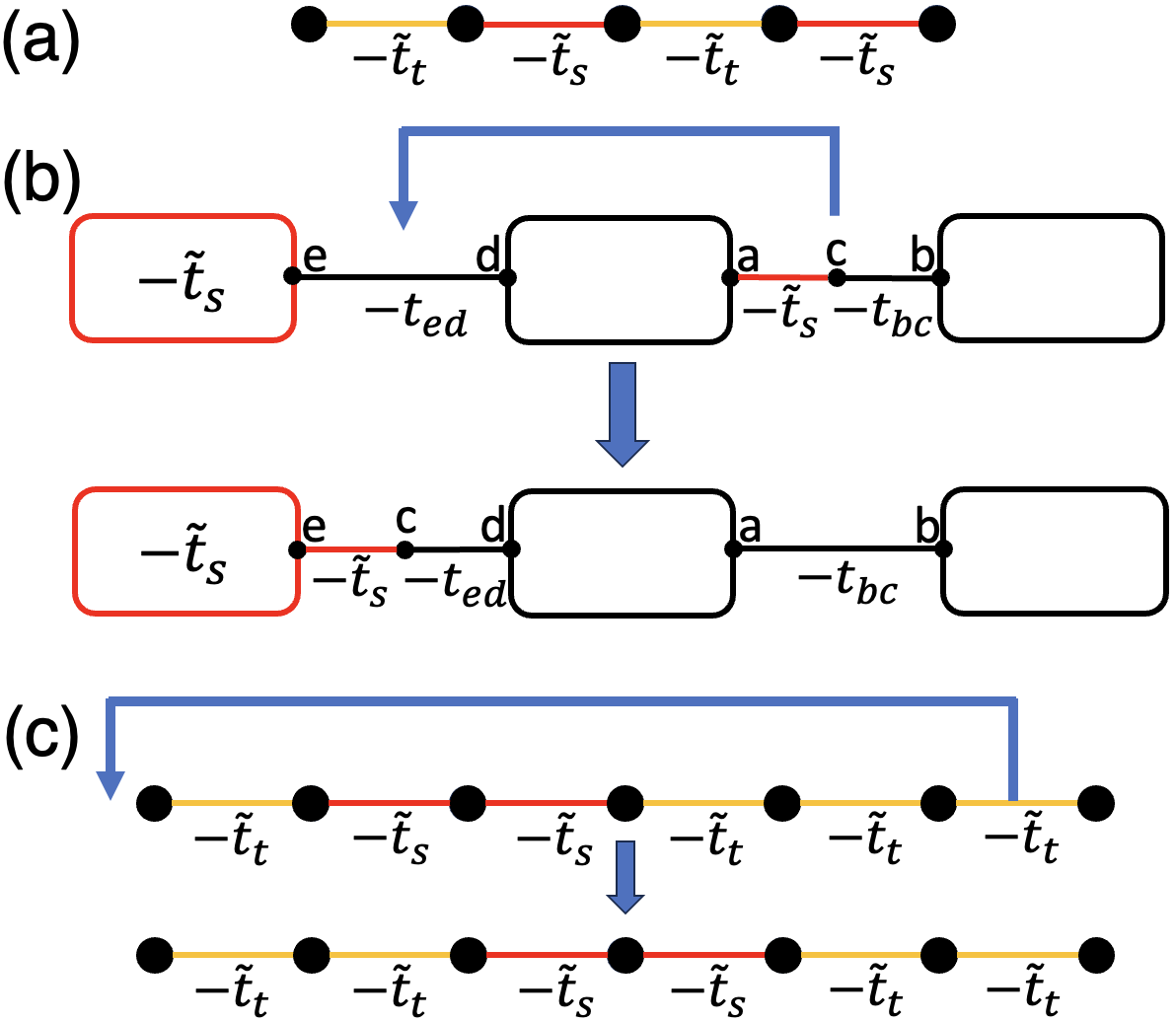}
\caption{Construction of the exact ground state in
each magnetization sector. (a) Mapping onto the basal one-dimensional chain;
singlet and triplet dimers generate the effective hoppings
$\widetilde t_{\rm s}$ and $\widetilde t_{\rm t}$ of
Lemma~\ref{lemm:1D_mapping} (red and yellow) [Eq.~\eqref{eq:eff_hoppings}].
(b) Connectivity step: the singlet bond $(a,c)$ is moved to
$(e,c)$, with $t_{ed},t_{bc}\in\{\widetilde t_{\rm t},
\widetilde t_{\rm s}\}$, preserving the number of singlets and
strictly lowering the energy [Eq.~\eqref{eq:connectivity_proof_main}].
(c) Center-positioning step: an off-center singlet
region is shifted toward the center by exchanging its two
tails.}
\label{fig:proof}
\end{figure}

\textit{Exact ground state in each magnetization sector.---}
Our central analytical result is:

\begin{theorem}[Exact ground state]
\label{thm:exact_gs}
Within every magnetization sector $M\geq0$, the minimum energy is attained by a configuration
with exactly $k=M$ triplet dimers and
$\ell_{\rm s}=L-M$ singlet dimers forming a single contiguous region.
For generic $t_2/t_1$, the triplet dimers are divided as equally
as possible between the two ends of the chain, centering the
singlet region to within one bond. At $t_2/t_1=t_{\ell_{\rm s}}$ and for $M\geq2$, all
positions of the singlet region separated from each physical end by at least one triplet dimer are degenerate; end-touching configurations lie higher in energy. At $h=0$, the singlet region fills the entire chain.
\end{theorem}

\begin{proof}
By Theorem~\ref{thm:sector}, the minimum over $k\geq M$ is attained at $k=M$. Using Lemma~\ref{lemm:1D_mapping}, the remaining proof has two steps
[Fig.~\ref{fig:proof}(b,c)].

\textit{Step 1: Connectivity.}
Two lemmas show that separated singlet regions cannot minimize
the energy.

\begin{lemma}[Rearrangement]
\label{lem:rearr}
Let $\mathcal C$ minimize the ground-state energy at fixed $k$,
with $\tilde t_t>0$ at $E=E_{\mathcal C}$, and let $\psi>0$ be the
ground state of its reduced chain
$\mathcal K_{\mathcal C}(E_{\mathcal C})$. If bonds $(i,i+1)$ and
$(j,j+1)$ carry hoppings $t_i<t_j$, then
$\psi_i\psi_{i+1}<\psi_j\psi_{j+1}$.
\end{lemma}
Otherwise, exchanging the two bond types at fixed $\psi$ would
lower the energy (SM~\cite{SM}, Sec.~\ref{app:connectivity}).

\begin{lemma}[Descent bond]
\label{lem:straddle}
Let $\psi\geq0$ be the gauged ground-state amplitudes of the reduced
chain, and let $\mathcal I_1$
and $\mathcal I_2$ be singlet regions separated by a triplet
region $\mathcal G$ (disjoint sets of basal sites),
with peak amplitudes
$\psi_j=\max_{\mathcal I_1}\psi
\leq\max_{\mathcal I_2}\psi=\psi_{j'}$.
Then there exists an oriented bond $(e,d)$ with
$e\in\mathcal I_2\cup\mathcal G$ such that
$\psi_e\geq\psi_j\geq\psi_d$.
\end{lemma}

These lemmas reduce the proof to the five-site rearrangement in
Fig.~\ref{fig:proof}(b), where the singlet bond $a\!-\!c$ is
replaced by $e\!-\!c$. Denoting the physical ground-state energies before and  after the rearrangement by $E_1$ and $E_2$, respectively, the rearrangement strictly lowers the total energy,
\begin{equation}
\Delta E_{\rm phys}\equiv E_1-E_2>0.
\end{equation}
This follows from the fixed-$E_1$ variational inequality 
\begin{align}
&\frac{\mathcal R_1-\mathcal R_2}{2}
=(\psi_e-\psi_c)\psi_d(\tilde t_s-t_{ed})+\nonumber\\
&(\psi_c-\psi_a)\psi_b(\tilde t_s-t_{bc})+
\tilde t_s(\psi_c-\psi_a)(\psi_c-\psi_b)\nonumber\\&-
\tilde t_s(\psi_e-\psi_c)(\psi_d-\psi_c)+\frac{\Delta U_{\rm b}}{2}
\geq0,
\label{eq:connectivity_proof_main}
\end{align}
together with the ground-state correspondence of Lemma~\ref{lemm:1D_mapping}. 
Here $\mathcal R_1$ and $\mathcal R_2$ denote the Rayleigh quotients of
the original and rearranged reduced 1D Hamiltonians,
both evaluated at $E=E_1$.
Each term in Eq. \eqref{eq:connectivity_proof_main} is non-negative: Lemma~\ref{lem:straddle} gives $\psi_e\ge\psi_c\ge\psi_d$, the peak
property gives $\psi_c\ge\psi_a$, Lemma~\ref{lem:rearr} gives $\psi_c\ge\psi_b$, and
$\tilde t_{s}\ge t_{ed},t_{bc}$ by construction. If $c$ is an end
site, the boundary potential contributes
$\Delta U_{\rm b}=U_{\rm b}(\psi_c^2-\psi_a^2)\geq0$; otherwise
$\Delta U_{\rm b}=0$. Equality is excluded (SM~\cite{SM}, Sec.~\ref{app:connectivity}):
it would require the reduced-chain eigenvalue
$\varepsilon(E)\equiv E-2t_1^2/E$ to equal
$-2\widetilde t_{\rm s}$, whereas the
open-chain spectral bound gives
\begin{equation}
\varepsilon(E)\geq
-2\widetilde t_{\rm s}
\cos\left(\frac{\pi}{L+2}\right)
>
-2\widetilde t_{\rm s}.
\end{equation}

\textit{Step 2: Center positioning.}
A tail-exchange argument lowers the energy whenever the connected
singlet region is off center. Iterating this operation divides the
triplet dimers as equally as possible between the two boundaries (SM~\cite{SM}, Sec.~\ref{app:center_pos}).
The exception occurs when
$\tilde t_t=0$: the hole has no tails in the triplet regions, the
tail exchange costs no energy, and all positions of the singlet
region that are separated from each physical end by at least one triplet dimer are degenerate.
\end{proof}

\textit{Magnetic confinement and penetration depth.---}
As $\ell_{\rm s}\to\infty$, the ground-state energy approaches
the bottom of the singlet band,
\begin{align}
E_\infty=-t_2-\sqrt{t_2^2+4t_1^2},
\label{eq:E_infty}
\end{align}
fixed by $E-2t_1^2/E=-2\tilde t_{\rm s}(E)$; matching at the
singlet--triplet boundary yields the finite-$\ell_{\rm s}$
correction and the tails (SM~\cite{SM}, Sec.~\ref{app:thermo_limit}). The hole wavefunction
is sinusoidal inside the singlet region and decays
outside it as $|\psi(x)|\propto e^{-x/L_{\rm loc}}$ [Fig.~\ref{fig:loc}(a--c)].
The exact penetration depth follows from
Eq.~\eqref{eq:Lloc_finite_energy} of the SM at the
finite-cage energy $E_\ell$. In the long-cage limit, with
$t\equiv t_2/t_1$ and
$\varepsilon\equiv|E_\infty|/t_1=t+\sqrt{t^2+4}$,
\begin{align}
L_{\rm loc}^{-1}(t)=
\mathrm{arccosh}\!\left[
\frac{\varepsilon(\varepsilon+2t)}
{4|t\varepsilon-1|}
\right],
\label{eq:Lloc_exact}
\end{align}
which for $L_{\rm loc}\gg1$ reduces to
\begin{align}
L_{\rm loc}(t)=
\begin{cases}
\sqrt{\dfrac{1-t\varepsilon}{4t\varepsilon}},
& t<1/\sqrt{6},\\[6pt]
\sqrt{\dfrac{t\varepsilon-1}{4}},
& t>1/\sqrt{6}.
\end{cases}
\label{eq:Lloc}
\end{align}
Figure~\ref{fig:loc}(d) compares Eq.~\eqref{eq:Lloc_exact} with the penetration depth
extracted from exact diagonalization of Eq.~\eqref{eq:H_TB} for $L=100$
and a centered singlet region with $\ell_{\rm s}=30$.
The agreement is within $1.7\,\%$ over $t_2/t_1\in[0.01,\,1]$ (99 of 100 sampled couplings have at least $5$ basal amplitudes above the numerical noise floor). Eq.~\eqref{eq:Lloc} is exact only for $t\to0$ and $t\to\infty$, where $L_{\rm loc}\to(8t)^{-1/2}$ and $L_{\rm loc}\to\sqrt{2}\,t/2$, respectively, and underestimates
$L_{\rm loc}$ near $t=1/\sqrt6$.

\begin{figure}[t]
\centering
\includegraphics[width=0.95\linewidth]{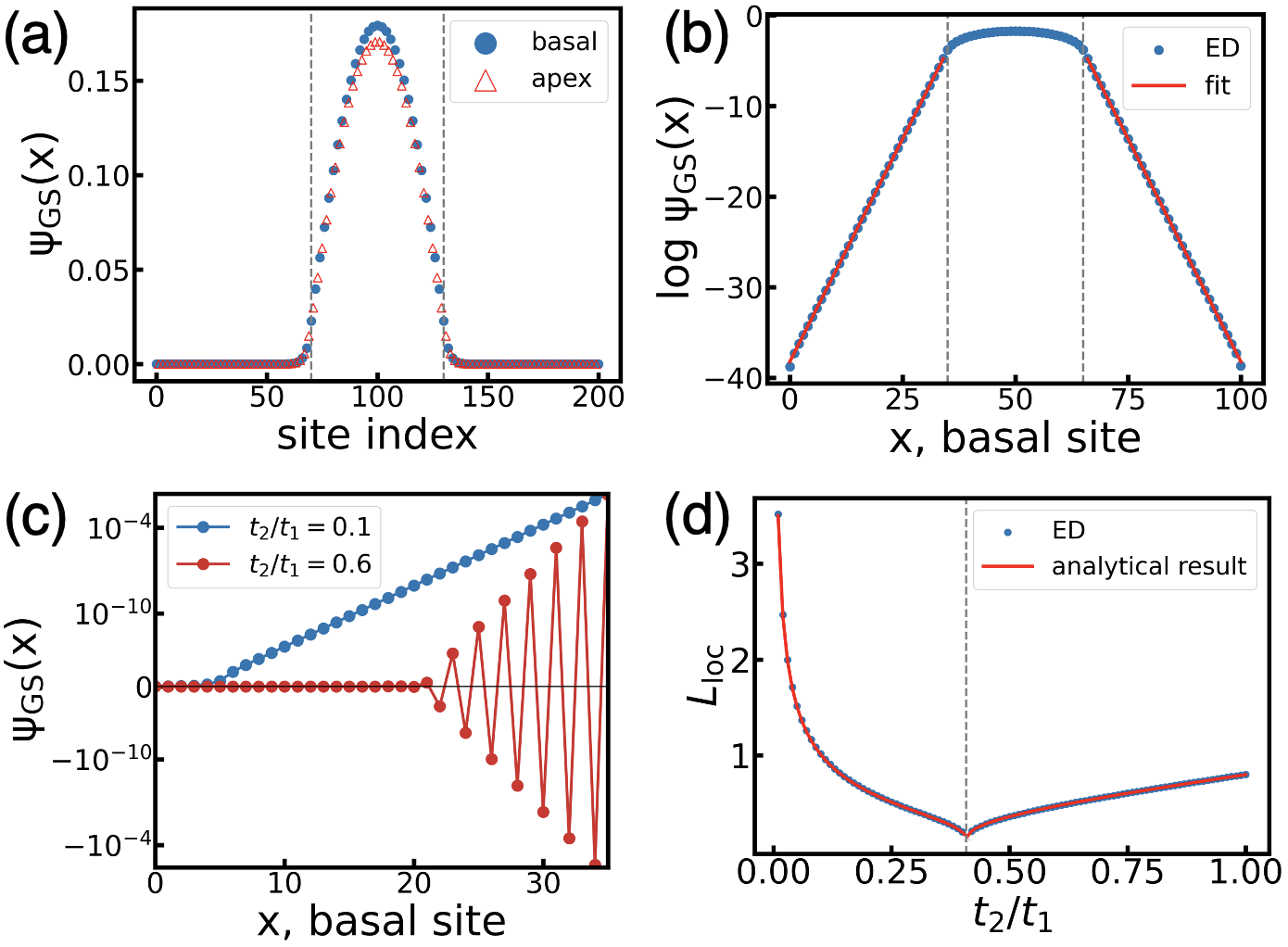}
\caption{(a) Ground-state wavefunction $\psi_\mathrm{GS}(x)$ on basal and apex sites for $L=100$ ($201$ sites), a singlet region of $\ell_{\rm s}=30$ bonds (dashed lines), and $t_2/t_1 = 0.1$. (b) $\log\psi_\mathrm{GS}(x)$ on the basal sublattice; red lines are exponential fits to the tails. (c) Left tail (basal sites $0$--$35$) on a signed log scale for $t_2/t_1=0.1$ and $0.6$, i.e., below and above the crossover $t_{30}\simeq0.409$ [Eq.~\eqref{eq:tells_main}]; above $t_{30}$ the decay is staggered.
(d) $L_{\rm loc}$ versus $t=t_2/t_1$: exact diagonalization for
$\ell_{\rm s}=30$ compared with Eq.~\eqref{eq:Lloc_exact}.
The dashed line marks $t=1/\sqrt6$.
}
\label{fig:loc}
\end{figure}

Both long-cage expressions vanish at $t\varepsilon=1$, i.e.,
$t=1/\sqrt6\simeq0.408$. This has a direct physical
meaning. Each triangle provides two paths between neighboring
basal sites: a direct basal hop and an apex-mediated hop. Across
a triplet dimer they cancel at
$E_\ast=-t_1^2/t_2$, making
the triplet region perfectly reflecting, whereas across a singlet dimer
they interfere constructively. For a finite cage, the tails therefore vanish
when the cage energy equals $E_\ast$, which occurs at the shifted coupling
$t_{\ell_{\rm s}}$ of Eq.~\eqref{eq:tells_main} ($t_{\ell_{\rm s}}\to1/\sqrt6$ as $\ell_{\rm s}\to\infty$).
For $t>t_{\ell_{\rm s}}$, the decay wavevector acquires a real part, producing an
alternating sign within the exponential envelope [Fig.~\ref{fig:loc}(c)].

While $L_{\rm loc}$ is the penetration depth
into the triplet region, the overall extent of the
state is set by $\ell_{\rm s}=L-M$. With $E_{\ell}$ the zero-field
ground-state energy at singlet length $\ell$ and
Zeeman contribution $-h(L-\ell)$, the field selects
\begin{equation}
\ell_{\rm s}(h)
\in
\underset{0\leq\ell\leq L}{\operatorname{argmin}}
\left[E_{\ell}+h\ell\right],
\qquad
\ell_{\rm s}(h)
\leq
\frac{E_0-E_\infty}{h}.
\label{eq:field_control}
\end{equation}
The bound follows from $E_{\ell}\geq E_\infty$ and shows
that $\ell_{\rm s}=O(1)$ at every fixed $h>0$ as
$L\to\infty$. Below the saturation field
$h_{\rm sat}$, the selected singlet region is nonempty
and confines the hole; above it,
$\ell_{\rm s}=0$ and the
hole moves in the uniform polarized background, where
its ground state is generally extended.
Comparing the minima at two fields
shows that $\ell_{\rm s}(h)$ is non-increasing in
$h$ (SM~\cite{SM}, Sec.~\ref{app:thermo_limit}): increasing the field shrinks the cage in discrete
steps, possibly skipping some values of $\ell_{\rm s}$
[Fig.~\ref{fig:confinement}].

\begin{figure}[t]
\centering
\includegraphics[width=\linewidth]{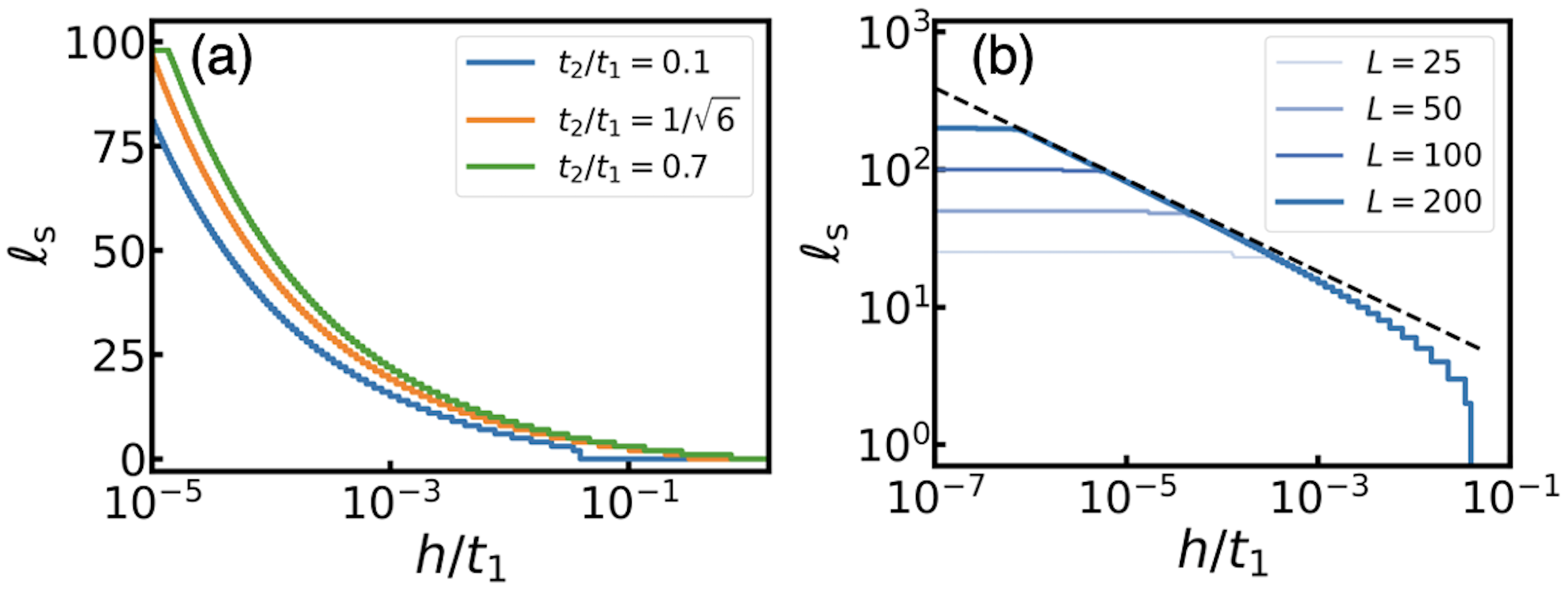}
\caption{Field-controlled confinement length.
(a) Singlet-region length $\ell_{\rm s}(h)$ from minimizing
$E_\ell+h\ell$, for $L=100$ and three couplings.
Each plateau fixes the magnetization deficit
$L-M=\ell_{\rm s}$ over a finite field interval; the corresponding
step $|\Delta m|=|\Delta\ell_{\rm s}|/L$ in magnetization per triangle ($m=M/L$)
vanishes as $L\to\infty$, while the cage length retains its discrete field dependence.
(b) Same at $t_2/t_1=0.1$ for several $L$: the
curves are $L$-independent except for the cutoff
$\ell_{\rm s}\leq L$ at $h=O(L^{-3})$. Dashed:
$(2t_1C/h)^{1/3}$ from minimizing
$E_\infty+t_1C/\ell^2+h\ell$, with $C$ from
Eq.~\eqref{eq:finite_size_bound} of the SM.}
    \label{fig:confinement}
\end{figure}

\textit{Compact localized states.---} At every finite singlet length
$\ell_{\rm s}$, an exact compact localized state exists at
\begin{equation}
\frac{t_2}{t_1}=t_{\ell_{\rm s}}
\equiv
\left[
2+4\cos\left(\frac{\pi}{\ell_{\rm s}+2}\right)
\right]^{-1/2}
\xrightarrow[\ell_{\rm s}\to\infty]{}
\frac{1}{\sqrt6},
\label{eq:tells_main}
\end{equation}
with basal amplitudes
\begin{equation}
\psi_i\propto
\sin\left[
\frac{\pi(i-a+1)}{\ell_{\rm s}+2}
\right]
\end{equation}
inside the singlet region $a\le i \le a+ \ell_s$, vanishing outside it. The
only additional support is on the $\ell_{\rm s}+2$ apex sites of
the singlet triangles and of the two flanking triplet triangles
(SM~\cite{SM}, Sec.~\ref{app:cls_exact}).
Its energy is independent of the position of the singlet region
as long as it is separated from each physical end by at least one triplet dimer.
The singlet region and its adjacent apex sites thus form a rigid quantum cage, generated by the
spin texture rather than by the bare lattice geometry~\cite{flat_band_1,flat_band_2,flat_band_3}, that can be shifted
along the chain without energy cost, giving a ground-state degeneracy that grows linearly
with system size.

\textit{Discussion.---}
A uniform magnetic field thus resizes an exact quantum cage for the doped hole: the field selects
the length of the singlet region, while destructive interference across the
surrounding triplet dimers forms its walls. The field-selected singlet region acts as an emergent,
magnetically tunable quantum dot whose walls are
generated by spin-dependent interference rather than by an
electrostatic potential.
DMRG calculations for two holes~\cite{White1992,TenPy} (SM~\cite{SM}, Sec.~\ref{app:two_hole}) indicate that the cage persists beyond a single hole: over most of the field range below saturation, both holes share one field-selected singlet region and their separation distribution is independent of system size, whereas at $h=0$ they are unbound.

Tunable sawtooth optical lattices and quantum-dot arrays provide
natural platforms for imaging this confinement and controlling
its size through spin imbalance or a Zeeman field. At finite $U$, superexchange
$J_{1,2}=4t_{1,2}^2/U$ modifies the kinetic problem.
At $J_1=J_2$, the zero-field
ground state of the projected $t$--$J$ model retains
its exact all-singlet dimer structure and its mapping
to a single spinless hole~\cite{kim2023exact};
the density term in
$J(\mathbf S_i\cdot\mathbf S_j-n_in_j/4)$ only adds
an effective on-site potential for the hole.
At finite magnetization, however,
exchange generally mixes singlet and triplet
configurations and makes the cage boundaries
dynamical. We therefore expect the cage
to evolve into a mobile magnetic polaron whose
bandwidth may be strongly suppressed for a large
cage, provided the field maintains a polarized
exterior. Establishing its dispersion and stability
at finite exchange remains an open question.

\textit{Acknowledgements.---} This work is supported by the Max Planck partner lab for quantum materials. Support for  C.D.B. came through Award DE-SC0018660 from the U.S. Department of Energy’s Office of Science, Basic Energy Sciences, Materials Sciences and Engineering Division. N.~P. acknowledges financial support from the Deutsche Forschungsgemeinschaft (DFG, German Research Foundation) through the Würzburg--Dresden Cluster of Excellence ctd.qmat -- Complexity, Topology and Dynamics in Quantum Matter (EXC 2147, project-id 390858490).

\nocite{White1992,Schollwock2011,TenPy,ogata-shiba}

\bibliography{references}

\begin{thebibliography}{46}%
\makeatletter
\providecommand \@ifxundefined [1]{%
 \@ifx{#1\undefined}
}%
\providecommand \@ifnum [1]{%
 \ifnum #1\expandafter \@firstoftwo
 \else \expandafter \@secondoftwo
 \fi
}%
\providecommand \@ifx [1]{%
 \ifx #1\expandafter \@firstoftwo
 \else \expandafter \@secondoftwo
 \fi
}%
\providecommand \natexlab [1]{#1}%
\providecommand \enquote  [1]{``#1''}%
\providecommand \bibnamefont  [1]{#1}%
\providecommand \bibfnamefont [1]{#1}%
\providecommand \citenamefont [1]{#1}%
\providecommand \href@noop [0]{\@secondoftwo}%
\providecommand \href [0]{\begingroup \@sanitize@url \@href}%
\providecommand \@href[1]{\@@startlink{#1}\@@href}%
\providecommand \@@href[1]{\endgroup#1\@@endlink}%
\providecommand \@sanitize@url [0]{\catcode `\\12\catcode `\$12\catcode
  `\&12\catcode `\#12\catcode `\^12\catcode `\_12\catcode `\%12\relax}%
\providecommand \@@startlink[1]{}%
\providecommand \@@endlink[0]{}%
\providecommand \url  [0]{\begingroup\@sanitize@url \@url }%
\providecommand \@url [1]{\endgroup\@href {#1}{\urlprefix }}%
\providecommand \urlprefix  [0]{URL }%
\providecommand \Eprint [0]{\href }%
\providecommand \doibase [0]{https://doi.org/}%
\providecommand \selectlanguage [0]{\@gobble}%
\providecommand \bibinfo  [0]{\@secondoftwo}%
\providecommand \bibfield  [0]{\@secondoftwo}%
\providecommand \translation [1]{[#1]}%
\providecommand \BibitemOpen [0]{}%
\providecommand \bibitemStop [0]{}%
\providecommand \bibitemNoStop [0]{.\EOS\space}%
\providecommand \EOS [0]{\spacefactor3000\relax}%
\providecommand \BibitemShut  [1]{\csname bibitem#1\endcsname}%
\let\auto@bib@innerbib\@empty
\bibitem [{\citenamefont {Anderson}\ \emph {et~al.}(1958)\citenamefont
  {Anderson} \emph {et~al.}}]{Anderson1958}%
  \BibitemOpen
  \bibfield  {author} {\bibinfo {author} {\bibfnamefont {P.~W.}\ \bibnamefont
  {Anderson}} \emph {et~al.},\ }\bibfield  {title} {\bibinfo {title} {Absence
  of diffusion in certain random lattices},\ }\href@noop {} {\bibfield
  {journal} {\bibinfo  {journal} {Physical review}\ }\textbf {\bibinfo {volume}
  {109}},\ \bibinfo {pages} {1492} (\bibinfo {year} {1958})}\BibitemShut
  {NoStop}%
\bibitem [{\citenamefont {Lee}\ and\ \citenamefont
  {Ramakrishnan}(1985)}]{Lee1985}%
  \BibitemOpen
  \bibfield  {author} {\bibinfo {author} {\bibfnamefont {P.~A.}\ \bibnamefont
  {Lee}}\ and\ \bibinfo {author} {\bibfnamefont {T.~V.}\ \bibnamefont
  {Ramakrishnan}},\ }\bibfield  {title} {\bibinfo {title} {Disordered
  electronic systems},\ }\href@noop {} {\bibfield  {journal} {\bibinfo
  {journal} {Reviews of modern physics}\ }\textbf {\bibinfo {volume} {57}},\
  \bibinfo {pages} {287} (\bibinfo {year} {1985})}\BibitemShut {NoStop}%
\bibitem [{\citenamefont {Evers}\ and\ \citenamefont
  {Mirlin}(2008)}]{Evers2008}%
  \BibitemOpen
  \bibfield  {author} {\bibinfo {author} {\bibfnamefont {F.}~\bibnamefont
  {Evers}}\ and\ \bibinfo {author} {\bibfnamefont {A.~D.}\ \bibnamefont
  {Mirlin}},\ }\bibfield  {title} {\bibinfo {title} {Anderson transitions},\
  }\href@noop {} {\bibfield  {journal} {\bibinfo  {journal} {Reviews of Modern
  Physics}\ }\textbf {\bibinfo {volume} {80}},\ \bibinfo {pages} {1355}
  (\bibinfo {year} {2008})}\BibitemShut {NoStop}%
\bibitem [{\citenamefont {Abrahams}\ \emph {et~al.}(1979)\citenamefont
  {Abrahams}, \citenamefont {Anderson}, \citenamefont {Licciardello},\ and\
  \citenamefont {Ramakrishnan}}]{abrahams1979}%
  \BibitemOpen
  \bibfield  {author} {\bibinfo {author} {\bibfnamefont {E.}~\bibnamefont
  {Abrahams}}, \bibinfo {author} {\bibfnamefont {P.~W.}\ \bibnamefont
  {Anderson}}, \bibinfo {author} {\bibfnamefont {D.~C.}\ \bibnamefont
  {Licciardello}},\ and\ \bibinfo {author} {\bibfnamefont {T.~V.}\ \bibnamefont
  {Ramakrishnan}},\ }\bibfield  {title} {\bibinfo {title} {Scaling theory of
  localization: Absence of quantum diffusion in two dimensions},\ }\href@noop
  {} {\bibfield  {journal} {\bibinfo  {journal} {Physical Review Letters}\
  }\textbf {\bibinfo {volume} {42}},\ \bibinfo {pages} {673} (\bibinfo {year}
  {1979})}\BibitemShut {NoStop}%
\bibitem [{\citenamefont {Anderson}(2010)}]{50years}%
  \BibitemOpen
  \bibfield  {author} {\bibinfo {author} {\bibfnamefont {P.}~\bibnamefont
  {Anderson}},\ }\bibfield  {title} {\bibinfo {title} {50 years of anderson
  localization},\ }\href@noop {} {\bibfield  {journal} {\bibinfo  {journal}
  {Edited by PW Anderson. Published by World Scientific Publishing Co. Pte.
  Ltd}\ } (\bibinfo {year} {2010})}\BibitemShut {NoStop}%
\bibitem [{\citenamefont {Schulz}\ \emph {et~al.}(2019)\citenamefont {Schulz},
  \citenamefont {Hooley}, \citenamefont {Moessner},\ and\ \citenamefont
  {Pollmann}}]{Schulz2019}%
  \BibitemOpen
  \bibfield  {author} {\bibinfo {author} {\bibfnamefont {M.}~\bibnamefont
  {Schulz}}, \bibinfo {author} {\bibfnamefont {C.}~\bibnamefont {Hooley}},
  \bibinfo {author} {\bibfnamefont {R.}~\bibnamefont {Moessner}},\ and\
  \bibinfo {author} {\bibfnamefont {F.}~\bibnamefont {Pollmann}},\ }\bibfield
  {title} {\bibinfo {title} {Stark many-body localization},\ }\href@noop {}
  {\bibfield  {journal} {\bibinfo  {journal} {Physical review letters}\
  }\textbf {\bibinfo {volume} {122}},\ \bibinfo {pages} {040606} (\bibinfo
  {year} {2019})}\BibitemShut {NoStop}%
\bibitem [{\citenamefont {van Nieuwenburg}\ \emph {et~al.}(2019)\citenamefont
  {van Nieuwenburg}, \citenamefont {Baum},\ and\ \citenamefont
  {Refael}}]{vanNieuwenburg2019}%
  \BibitemOpen
  \bibfield  {author} {\bibinfo {author} {\bibfnamefont {E.}~\bibnamefont {van
  Nieuwenburg}}, \bibinfo {author} {\bibfnamefont {Y.}~\bibnamefont {Baum}},\
  and\ \bibinfo {author} {\bibfnamefont {G.}~\bibnamefont {Refael}},\
  }\bibfield  {title} {\bibinfo {title} {From bloch oscillations to many-body
  localization in clean interacting systems},\ }\href@noop {} {\bibfield
  {journal} {\bibinfo  {journal} {Proceedings of the National Academy of
  Sciences}\ }\textbf {\bibinfo {volume} {116}},\ \bibinfo {pages} {9269}
  (\bibinfo {year} {2019})}\BibitemShut {NoStop}%
\bibitem [{\citenamefont {Scherg}\ \emph {et~al.}(2021)\citenamefont {Scherg},
  \citenamefont {Kohlert}, \citenamefont {Sala}, \citenamefont {Pollmann},
  \citenamefont {Hebbe~Madhusudhana}, \citenamefont {Bloch},\ and\
  \citenamefont {Aidelsburger}}]{Scherg2021}%
  \BibitemOpen
  \bibfield  {author} {\bibinfo {author} {\bibfnamefont {S.}~\bibnamefont
  {Scherg}}, \bibinfo {author} {\bibfnamefont {T.}~\bibnamefont {Kohlert}},
  \bibinfo {author} {\bibfnamefont {P.}~\bibnamefont {Sala}}, \bibinfo {author}
  {\bibfnamefont {F.}~\bibnamefont {Pollmann}}, \bibinfo {author}
  {\bibfnamefont {B.}~\bibnamefont {Hebbe~Madhusudhana}}, \bibinfo {author}
  {\bibfnamefont {I.}~\bibnamefont {Bloch}},\ and\ \bibinfo {author}
  {\bibfnamefont {M.}~\bibnamefont {Aidelsburger}},\ }\bibfield  {title}
  {\bibinfo {title} {Observing non-ergodicity due to kinetic constraints in
  tilted fermi-hubbard chains},\ }\href@noop {} {\bibfield  {journal} {\bibinfo
   {journal} {Nature Communications}\ }\textbf {\bibinfo {volume} {12}},\
  \bibinfo {pages} {4490} (\bibinfo {year} {2021})}\BibitemShut {NoStop}%
\bibitem [{\citenamefont {Sala}\ \emph {et~al.}(2020)\citenamefont {Sala},
  \citenamefont {Rakovszky}, \citenamefont {Verresen}, \citenamefont {Knap},\
  and\ \citenamefont {Pollmann}}]{Sala2020}%
  \BibitemOpen
  \bibfield  {author} {\bibinfo {author} {\bibfnamefont {P.}~\bibnamefont
  {Sala}}, \bibinfo {author} {\bibfnamefont {T.}~\bibnamefont {Rakovszky}},
  \bibinfo {author} {\bibfnamefont {R.}~\bibnamefont {Verresen}}, \bibinfo
  {author} {\bibfnamefont {M.}~\bibnamefont {Knap}},\ and\ \bibinfo {author}
  {\bibfnamefont {F.}~\bibnamefont {Pollmann}},\ }\bibfield  {title} {\bibinfo
  {title} {Ergodicity breaking arising from hilbert space fragmentation in
  dipole-conserving hamiltonians},\ }\href@noop {} {\bibfield  {journal}
  {\bibinfo  {journal} {Physical Review X}\ }\textbf {\bibinfo {volume} {10}},\
  \bibinfo {pages} {011047} (\bibinfo {year} {2020})}\BibitemShut {NoStop}%
\bibitem [{\citenamefont {Khemani}\ \emph {et~al.}(2020)\citenamefont
  {Khemani}, \citenamefont {Hermele},\ and\ \citenamefont
  {Nandkishore}}]{Khemani2020}%
  \BibitemOpen
  \bibfield  {author} {\bibinfo {author} {\bibfnamefont {V.}~\bibnamefont
  {Khemani}}, \bibinfo {author} {\bibfnamefont {M.}~\bibnamefont {Hermele}},\
  and\ \bibinfo {author} {\bibfnamefont {R.}~\bibnamefont {Nandkishore}},\
  }\bibfield  {title} {\bibinfo {title} {Localization from hilbert space
  shattering: From theory to physical realizations},\ }\href@noop {} {\bibfield
   {journal} {\bibinfo  {journal} {Phys. Rev. B}\ }\textbf {\bibinfo {volume}
  {101}},\ \bibinfo {pages} {174204} (\bibinfo {year} {2020})}\BibitemShut
  {NoStop}%
\bibitem [{\citenamefont {Moudgalya}\ \emph {et~al.}(2022)\citenamefont
  {Moudgalya}, \citenamefont {Bernevig},\ and\ \citenamefont
  {Regnault}}]{Moudgalya2022}%
  \BibitemOpen
  \bibfield  {author} {\bibinfo {author} {\bibfnamefont {S.}~\bibnamefont
  {Moudgalya}}, \bibinfo {author} {\bibfnamefont {B.~A.}\ \bibnamefont
  {Bernevig}},\ and\ \bibinfo {author} {\bibfnamefont {N.}~\bibnamefont
  {Regnault}},\ }\bibfield  {title} {\bibinfo {title} {Quantum many-body scars
  and hilbert space fragmentation: A review of exact results},\ }\href@noop {}
  {\bibfield  {journal} {\bibinfo  {journal} {Reports on Progress in Physics}\
  }\textbf {\bibinfo {volume} {85}},\ \bibinfo {pages} {086501} (\bibinfo
  {year} {2022})}\BibitemShut {NoStop}%
\bibitem [{\citenamefont {Rakovszky}\ \emph {et~al.}(2020)\citenamefont
  {Rakovszky}, \citenamefont {Sala}, \citenamefont {Verresen}, \citenamefont
  {Knap},\ and\ \citenamefont {Pollmann}}]{Rakovszky2020}%
  \BibitemOpen
  \bibfield  {author} {\bibinfo {author} {\bibfnamefont {T.}~\bibnamefont
  {Rakovszky}}, \bibinfo {author} {\bibfnamefont {P.}~\bibnamefont {Sala}},
  \bibinfo {author} {\bibfnamefont {R.}~\bibnamefont {Verresen}}, \bibinfo
  {author} {\bibfnamefont {M.}~\bibnamefont {Knap}},\ and\ \bibinfo {author}
  {\bibfnamefont {F.}~\bibnamefont {Pollmann}},\ }\bibfield  {title} {\bibinfo
  {title} {Statistical localization: From strong fragmentation to strong edge
  modes},\ }\href@noop {} {\bibfield  {journal} {\bibinfo  {journal} {Physical
  Review B}\ }\textbf {\bibinfo {volume} {101}},\ \bibinfo {pages} {125126}
  (\bibinfo {year} {2020})}\BibitemShut {NoStop}%
\bibitem [{\citenamefont {Smith}\ \emph {et~al.}(2017)\citenamefont {Smith},
  \citenamefont {Knolle}, \citenamefont {Kovrizhin},\ and\ \citenamefont
  {Moessner}}]{Smith2017}%
  \BibitemOpen
  \bibfield  {author} {\bibinfo {author} {\bibfnamefont {A.}~\bibnamefont
  {Smith}}, \bibinfo {author} {\bibfnamefont {J.}~\bibnamefont {Knolle}},
  \bibinfo {author} {\bibfnamefont {D.~L.}\ \bibnamefont {Kovrizhin}},\ and\
  \bibinfo {author} {\bibfnamefont {R.}~\bibnamefont {Moessner}},\ }\bibfield
  {title} {\bibinfo {title} {Disorder-free localization},\ }\href@noop {}
  {\bibfield  {journal} {\bibinfo  {journal} {Physical review letters}\
  }\textbf {\bibinfo {volume} {118}},\ \bibinfo {pages} {266601} (\bibinfo
  {year} {2017})}\BibitemShut {NoStop}%
\bibitem [{\citenamefont {Brenes}\ \emph {et~al.}(2018)\citenamefont {Brenes},
  \citenamefont {Dalmonte}, \citenamefont {Heyl},\ and\ \citenamefont
  {Scardicchio}}]{Brenes2018}%
  \BibitemOpen
  \bibfield  {author} {\bibinfo {author} {\bibfnamefont {M.}~\bibnamefont
  {Brenes}}, \bibinfo {author} {\bibfnamefont {M.}~\bibnamefont {Dalmonte}},
  \bibinfo {author} {\bibfnamefont {M.}~\bibnamefont {Heyl}},\ and\ \bibinfo
  {author} {\bibfnamefont {A.}~\bibnamefont {Scardicchio}},\ }\bibfield
  {title} {\bibinfo {title} {Many-body localization dynamics from gauge
  invariance},\ }\href@noop {} {\bibfield  {journal} {\bibinfo  {journal}
  {Physical review letters}\ }\textbf {\bibinfo {volume} {120}},\ \bibinfo
  {pages} {030601} (\bibinfo {year} {2018})}\BibitemShut {NoStop}%
\bibitem [{\citenamefont {Schiulaz}\ \emph {et~al.}(2015)\citenamefont
  {Schiulaz}, \citenamefont {Silva},\ and\ \citenamefont
  {M{\"u}ller}}]{Schiulaz2015}%
  \BibitemOpen
  \bibfield  {author} {\bibinfo {author} {\bibfnamefont {M.}~\bibnamefont
  {Schiulaz}}, \bibinfo {author} {\bibfnamefont {A.}~\bibnamefont {Silva}},\
  and\ \bibinfo {author} {\bibfnamefont {M.}~\bibnamefont {M{\"u}ller}},\
  }\bibfield  {title} {\bibinfo {title} {Dynamics in many-body localized
  quantum systems without disorder},\ }\href@noop {} {\bibfield  {journal}
  {\bibinfo  {journal} {Physical Review B}\ }\textbf {\bibinfo {volume} {91}},\
  \bibinfo {pages} {184202} (\bibinfo {year} {2015})}\BibitemShut {NoStop}%
\bibitem [{\citenamefont {De~Roeck}\ and\ \citenamefont
  {Huveneers}(2014)}]{DeRoeck2014}%
  \BibitemOpen
  \bibfield  {author} {\bibinfo {author} {\bibfnamefont {W.}~\bibnamefont
  {De~Roeck}}\ and\ \bibinfo {author} {\bibfnamefont {F.}~\bibnamefont
  {Huveneers}},\ }\bibfield  {title} {\bibinfo {title} {Scenario for
  delocalization in translation-invariant systems},\ }\href@noop {} {\bibfield
  {journal} {\bibinfo  {journal} {Physical Review B}\ }\textbf {\bibinfo
  {volume} {90}},\ \bibinfo {pages} {165137} (\bibinfo {year}
  {2014})}\BibitemShut {NoStop}%
\bibitem [{\citenamefont {Yao}\ \emph {et~al.}(2016)\citenamefont {Yao},
  \citenamefont {Laumann}, \citenamefont {Cirac}, \citenamefont {Lukin},\ and\
  \citenamefont {Moore}}]{Yao2016}%
  \BibitemOpen
  \bibfield  {author} {\bibinfo {author} {\bibfnamefont {N.}~\bibnamefont
  {Yao}}, \bibinfo {author} {\bibfnamefont {C.}~\bibnamefont {Laumann}},
  \bibinfo {author} {\bibfnamefont {J.~I.}\ \bibnamefont {Cirac}}, \bibinfo
  {author} {\bibfnamefont {M.~D.}\ \bibnamefont {Lukin}},\ and\ \bibinfo
  {author} {\bibfnamefont {J.}~\bibnamefont {Moore}},\ }\bibfield  {title}
  {\bibinfo {title} {Quasi-many-body localization in translation-invariant
  systems},\ }\href@noop {} {\bibfield  {journal} {\bibinfo  {journal}
  {Physical review letters}\ }\textbf {\bibinfo {volume} {117}},\ \bibinfo
  {pages} {240601} (\bibinfo {year} {2016})}\BibitemShut {NoStop}%
\bibitem [{\citenamefont {Basko}\ \emph {et~al.}(2006)\citenamefont {Basko},
  \citenamefont {Aleiner},\ and\ \citenamefont {Altshuler}}]{Basko2006}%
  \BibitemOpen
  \bibfield  {author} {\bibinfo {author} {\bibfnamefont {D.~M.}\ \bibnamefont
  {Basko}}, \bibinfo {author} {\bibfnamefont {I.~L.}\ \bibnamefont {Aleiner}},\
  and\ \bibinfo {author} {\bibfnamefont {B.~L.}\ \bibnamefont {Altshuler}},\
  }\bibfield  {title} {\bibinfo {title} {Metal--insulator transition in a
  weakly interacting many-electron system with localized single-particle
  states},\ }\href@noop {} {\bibfield  {journal} {\bibinfo  {journal} {Annals
  of physics}\ }\textbf {\bibinfo {volume} {321}},\ \bibinfo {pages} {1126}
  (\bibinfo {year} {2006})}\BibitemShut {NoStop}%
\bibitem [{\citenamefont {Gornyi}\ \emph {et~al.}(2005)\citenamefont {Gornyi},
  \citenamefont {Mirlin},\ and\ \citenamefont {Polyakov}}]{Gornyi2005}%
  \BibitemOpen
  \bibfield  {author} {\bibinfo {author} {\bibfnamefont {I.~V.}\ \bibnamefont
  {Gornyi}}, \bibinfo {author} {\bibfnamefont {A.~D.}\ \bibnamefont {Mirlin}},\
  and\ \bibinfo {author} {\bibfnamefont {D.~G.}\ \bibnamefont {Polyakov}},\
  }\bibfield  {title} {\bibinfo {title} {Interacting electrons in disordered
  wires: Anderson localization and low-t transport},\ }\href@noop {} {\bibfield
   {journal} {\bibinfo  {journal} {Physical review letters}\ }\textbf {\bibinfo
  {volume} {95}},\ \bibinfo {pages} {206603} (\bibinfo {year}
  {2005})}\BibitemShut {NoStop}%
\bibitem [{\citenamefont {Nandkishore}\ and\ \citenamefont
  {Huse}(2015)}]{Nandkishore2015}%
  \BibitemOpen
  \bibfield  {author} {\bibinfo {author} {\bibfnamefont {R.}~\bibnamefont
  {Nandkishore}}\ and\ \bibinfo {author} {\bibfnamefont {D.~A.}\ \bibnamefont
  {Huse}},\ }\bibfield  {title} {\bibinfo {title} {Many-body localization and
  thermalization in quantum statistical mechanics},\ }\href@noop {} {\bibfield
  {journal} {\bibinfo  {journal} {Annu. Rev. Condens. Matter Phys.}\ }\textbf
  {\bibinfo {volume} {6}},\ \bibinfo {pages} {15} (\bibinfo {year}
  {2015})}\BibitemShut {NoStop}%
\bibitem [{\citenamefont {Abanin}\ \emph {et~al.}(2019)\citenamefont {Abanin},
  \citenamefont {Altman}, \citenamefont {Bloch},\ and\ \citenamefont
  {Serbyn}}]{Abanin2019}%
  \BibitemOpen
  \bibfield  {author} {\bibinfo {author} {\bibfnamefont {D.~A.}\ \bibnamefont
  {Abanin}}, \bibinfo {author} {\bibfnamefont {E.}~\bibnamefont {Altman}},
  \bibinfo {author} {\bibfnamefont {I.}~\bibnamefont {Bloch}},\ and\ \bibinfo
  {author} {\bibfnamefont {M.}~\bibnamefont {Serbyn}},\ }\bibfield  {title}
  {\bibinfo {title} {Colloquium: Many-body localization, thermalization, and
  entanglement},\ }\href@noop {} {\bibfield  {journal} {\bibinfo  {journal}
  {Reviews of Modern Physics}\ }\textbf {\bibinfo {volume} {91}},\ \bibinfo
  {pages} {021001} (\bibinfo {year} {2019})}\BibitemShut {NoStop}%
\bibitem [{\citenamefont {Sutherland}(1986)}]{flat_band_1}%
  \BibitemOpen
  \bibfield  {author} {\bibinfo {author} {\bibfnamefont {B.}~\bibnamefont
  {Sutherland}},\ }\bibfield  {title} {\bibinfo {title} {Localization of
  electronic wave functions due to local topology},\ }\href@noop {} {\bibfield
  {journal} {\bibinfo  {journal} {Physical Review B}\ }\textbf {\bibinfo
  {volume} {34}},\ \bibinfo {pages} {5208} (\bibinfo {year}
  {1986})}\BibitemShut {NoStop}%
\bibitem [{\citenamefont {Leykam}\ \emph {et~al.}(2018)\citenamefont {Leykam},
  \citenamefont {Andreanov},\ and\ \citenamefont {Flach}}]{Leykam2018}%
  \BibitemOpen
  \bibfield  {author} {\bibinfo {author} {\bibfnamefont {D.}~\bibnamefont
  {Leykam}}, \bibinfo {author} {\bibfnamefont {A.}~\bibnamefont {Andreanov}},\
  and\ \bibinfo {author} {\bibfnamefont {S.}~\bibnamefont {Flach}},\ }\bibfield
   {title} {\bibinfo {title} {Artificial flat band systems: from lattice models
  to experiments},\ }\href@noop {} {\bibfield  {journal} {\bibinfo  {journal}
  {Advances in Physics: X}\ }\textbf {\bibinfo {volume} {3}},\ \bibinfo {pages}
  {1473052} (\bibinfo {year} {2018})}\BibitemShut {NoStop}%
\bibitem [{\citenamefont {Danieli}\ \emph {et~al.}(2021)\citenamefont
  {Danieli}, \citenamefont {Andreanov}, \citenamefont {Mithun},\ and\
  \citenamefont {Flach}}]{Danieli2021}%
  \BibitemOpen
  \bibfield  {author} {\bibinfo {author} {\bibfnamefont {C.}~\bibnamefont
  {Danieli}}, \bibinfo {author} {\bibfnamefont {A.}~\bibnamefont {Andreanov}},
  \bibinfo {author} {\bibfnamefont {T.}~\bibnamefont {Mithun}},\ and\ \bibinfo
  {author} {\bibfnamefont {S.}~\bibnamefont {Flach}},\ }\bibfield  {title}
  {\bibinfo {title} {Quantum caging in interacting many-body all-bands-flat
  lattices},\ }\href {https://doi.org/10.1103/PhysRevB.104.085132} {\bibfield
  {journal} {\bibinfo  {journal} {Phys. Rev. B}\ }\textbf {\bibinfo {volume}
  {104}},\ \bibinfo {pages} {085132} (\bibinfo {year} {2021})}\BibitemShut
  {NoStop}%
\bibitem [{\citenamefont {Nagaoka}(1966)}]{Nagaoka_FM}%
  \BibitemOpen
  \bibfield  {author} {\bibinfo {author} {\bibfnamefont {Y.}~\bibnamefont
  {Nagaoka}},\ }\bibfield  {title} {\bibinfo {title} {Ferromagnetism in a
  narrow, almost half-filled s band},\ }\href@noop {} {\bibfield  {journal}
  {\bibinfo  {journal} {Physical Review}\ }\textbf {\bibinfo {volume} {147}},\
  \bibinfo {pages} {392} (\bibinfo {year} {1966})}\BibitemShut {NoStop}%
\bibitem [{\citenamefont {Tasaki}(1989)}]{tasaki1989extension}%
  \BibitemOpen
  \bibfield  {author} {\bibinfo {author} {\bibfnamefont {H.}~\bibnamefont
  {Tasaki}},\ }\bibfield  {title} {\bibinfo {title} {Extension of nagaoka’s
  theorem on the large-u hubbard model},\ }\href@noop {} {\bibfield  {journal}
  {\bibinfo  {journal} {Physical Review B}\ }\textbf {\bibinfo {volume} {40}},\
  \bibinfo {pages} {9192} (\bibinfo {year} {1989})}\BibitemShut {NoStop}%
\bibitem [{\citenamefont {Tasaki}(1998)}]{tasaki1998nagaoka}%
  \BibitemOpen
  \bibfield  {author} {\bibinfo {author} {\bibfnamefont {H.}~\bibnamefont
  {Tasaki}},\ }\bibfield  {title} {\bibinfo {title} {From nagaoka's
  ferromagnetism to flat-band ferromagnetism and beyond: An introduction to
  ferromagnetism in the hubbard model},\ }\href@noop {} {\bibfield  {journal}
  {\bibinfo  {journal} {Progress of theoretical physics}\ }\textbf {\bibinfo
  {volume} {99}},\ \bibinfo {pages} {489} (\bibinfo {year} {1998})}\BibitemShut
  {NoStop}%
\bibitem [{\citenamefont {Haerter}\ and\ \citenamefont
  {Shastry}(2005)}]{shastry_kinetic}%
  \BibitemOpen
  \bibfield  {author} {\bibinfo {author} {\bibfnamefont {J.~O.}\ \bibnamefont
  {Haerter}}\ and\ \bibinfo {author} {\bibfnamefont {B.~S.}\ \bibnamefont
  {Shastry}},\ }\bibfield  {title} {\bibinfo {title} {Kinetic
  antiferromagnetism in the triangular lattice},\ }\href@noop {} {\bibfield
  {journal} {\bibinfo  {journal} {Physical review letters}\ }\textbf {\bibinfo
  {volume} {95}},\ \bibinfo {pages} {087202} (\bibinfo {year}
  {2005})}\BibitemShut {NoStop}%
\bibitem [{\citenamefont {Pollmann}\ \emph {et~al.}(2008)\citenamefont
  {Pollmann}, \citenamefont {Fulde},\ and\ \citenamefont
  {Shtengel}}]{Pollmann2008}%
  \BibitemOpen
  \bibfield  {author} {\bibinfo {author} {\bibfnamefont {F.}~\bibnamefont
  {Pollmann}}, \bibinfo {author} {\bibfnamefont {P.}~\bibnamefont {Fulde}},\
  and\ \bibinfo {author} {\bibfnamefont {K.}~\bibnamefont {Shtengel}},\
  }\bibfield  {title} {\bibinfo {title} {Kinetic ferromagnetism on a kagome
  lattice},\ }\href@noop {} {\bibfield  {journal} {\bibinfo  {journal}
  {Physical review letters}\ }\textbf {\bibinfo {volume} {100}},\ \bibinfo
  {pages} {136404} (\bibinfo {year} {2008})}\BibitemShut {NoStop}%
\bibitem [{\citenamefont {Sposetti}\ \emph {et~al.}(2014)\citenamefont
  {Sposetti}, \citenamefont {Bravo}, \citenamefont {Trumper}, \citenamefont
  {Gazza},\ and\ \citenamefont {Manuel}}]{Sposetti2014}%
  \BibitemOpen
  \bibfield  {author} {\bibinfo {author} {\bibfnamefont {C.~N.}\ \bibnamefont
  {Sposetti}}, \bibinfo {author} {\bibfnamefont {B.}~\bibnamefont {Bravo}},
  \bibinfo {author} {\bibfnamefont {A.~E.}\ \bibnamefont {Trumper}}, \bibinfo
  {author} {\bibfnamefont {C.~J.}\ \bibnamefont {Gazza}},\ and\ \bibinfo
  {author} {\bibfnamefont {L.~O.}\ \bibnamefont {Manuel}},\ }\bibfield  {title}
  {\bibinfo {title} {Classical antiferromagnetism in kinetically frustrated
  electronic models},\ }\href {https://doi.org/10.1103/PhysRevLett.112.187204}
  {\bibfield  {journal} {\bibinfo  {journal} {Phys. Rev. Lett.}\ }\textbf
  {\bibinfo {volume} {112}},\ \bibinfo {pages} {187204} (\bibinfo {year}
  {2014})}\BibitemShut {NoStop}%
\bibitem [{\citenamefont {Kim}(2023)}]{kim2023exact}%
  \BibitemOpen
  \bibfield  {author} {\bibinfo {author} {\bibfnamefont {K.-S.}\ \bibnamefont
  {Kim}},\ }\bibfield  {title} {\bibinfo {title} {Exact hole-induced
  resonating-valence-bond ground state in certain $u=\infty$ hubbard models},\
  }\href@noop {} {\bibfield  {journal} {\bibinfo  {journal} {Physical Review
  B}\ }\textbf {\bibinfo {volume} {107}},\ \bibinfo {pages} {L140401} (\bibinfo
  {year} {2023})}\BibitemShut {NoStop}%
\bibitem [{\citenamefont {Zhang}\ and\ \citenamefont
  {Glittum}(2026)}]{sawtooth_zhang}%
  \BibitemOpen
  \bibfield  {author} {\bibinfo {author} {\bibfnamefont {Z.}~\bibnamefont
  {Zhang}}\ and\ \bibinfo {author} {\bibfnamefont {C.}~\bibnamefont
  {Glittum}},\ }\bibfield  {title} {\bibinfo {title} {Resonating valence bond
  ground states on corner-sharing simplices},\ }\href@noop {} {\bibfield
  {journal} {\bibinfo  {journal} {SciPost Physics}\ }\textbf {\bibinfo {volume}
  {20}},\ \bibinfo {pages} {121} (\bibinfo {year} {2026})}\BibitemShut
  {NoStop}%
\bibitem [{\citenamefont {Glittum}\ \emph {et~al.}(2025)\citenamefont
  {Glittum}, \citenamefont {{\v{S}}trkalj}, \citenamefont {Prabhakaran},
  \citenamefont {Goddard}, \citenamefont {Batista},\ and\ \citenamefont
  {Castelnovo}}]{kinetic_frustration_pyrochlore}%
  \BibitemOpen
  \bibfield  {author} {\bibinfo {author} {\bibfnamefont {C.}~\bibnamefont
  {Glittum}}, \bibinfo {author} {\bibfnamefont {A.}~\bibnamefont
  {{\v{S}}trkalj}}, \bibinfo {author} {\bibfnamefont {D.}~\bibnamefont
  {Prabhakaran}}, \bibinfo {author} {\bibfnamefont {P.~A.}\ \bibnamefont
  {Goddard}}, \bibinfo {author} {\bibfnamefont {C.~D.}\ \bibnamefont
  {Batista}},\ and\ \bibinfo {author} {\bibfnamefont {C.}~\bibnamefont
  {Castelnovo}},\ }\bibfield  {title} {\bibinfo {title} {A resonant valence
  bond spin liquid in the dilute limit of doped frustrated mott insulators},\
  }\href@noop {} {\bibfield  {journal} {\bibinfo  {journal} {Nature Physics}\
  }\textbf {\bibinfo {volume} {21}},\ \bibinfo {pages} {1211} (\bibinfo {year}
  {2025})}\BibitemShut {NoStop}%
\bibitem [{\citenamefont {Hubbard}(1963)}]{Hubbard1963}%
  \BibitemOpen
  \bibfield  {author} {\bibinfo {author} {\bibfnamefont {J.}~\bibnamefont
  {Hubbard}},\ }\bibfield  {title} {\bibinfo {title} {Electron correlations in
  narrow energy bands},\ }\href@noop {} {\bibfield  {journal} {\bibinfo
  {journal} {Proceedings of the Royal Society of London. Series A. Mathematical
  and Physical Sciences}\ }\textbf {\bibinfo {volume} {276}},\ \bibinfo {pages}
  {238} (\bibinfo {year} {1963})}\BibitemShut {NoStop}%
\bibitem [{\citenamefont {Imada}\ \emph {et~al.}(1998)\citenamefont {Imada},
  \citenamefont {Fujimori},\ and\ \citenamefont {Tokura}}]{Imada1998}%
  \BibitemOpen
  \bibfield  {author} {\bibinfo {author} {\bibfnamefont {M.}~\bibnamefont
  {Imada}}, \bibinfo {author} {\bibfnamefont {A.}~\bibnamefont {Fujimori}},\
  and\ \bibinfo {author} {\bibfnamefont {Y.}~\bibnamefont {Tokura}},\
  }\bibfield  {title} {\bibinfo {title} {Metal-insulator transitions},\
  }\href@noop {} {\bibfield  {journal} {\bibinfo  {journal} {Reviews of modern
  physics}\ }\textbf {\bibinfo {volume} {70}},\ \bibinfo {pages} {1039}
  (\bibinfo {year} {1998})}\BibitemShut {NoStop}%
\bibitem [{\citenamefont {Lieb}\ and\ \citenamefont {Wu}(1968)}]{LiebWu1968}%
  \BibitemOpen
  \bibfield  {author} {\bibinfo {author} {\bibfnamefont {E.~H.}\ \bibnamefont
  {Lieb}}\ and\ \bibinfo {author} {\bibfnamefont {F.-Y.}\ \bibnamefont {Wu}},\
  }\bibfield  {title} {\bibinfo {title} {Absence of mott transition in an exact
  solution of the short-range, one-band model in one dimension},\ }\href@noop
  {} {\bibfield  {journal} {\bibinfo  {journal} {Physical Review Letters}\
  }\textbf {\bibinfo {volume} {20}},\ \bibinfo {pages} {1445} (\bibinfo {year}
  {1968})}\BibitemShut {NoStop}%
\bibitem [{\citenamefont {Sen}\ \emph {et~al.}(1996)\citenamefont {Sen},
  \citenamefont {Shastry}, \citenamefont {Walstedt},\ and\ \citenamefont
  {Cava}}]{Sen1996}%
  \BibitemOpen
  \bibfield  {author} {\bibinfo {author} {\bibfnamefont {D.}~\bibnamefont
  {Sen}}, \bibinfo {author} {\bibfnamefont {B.~S.}\ \bibnamefont {Shastry}},
  \bibinfo {author} {\bibfnamefont {R.~E.}\ \bibnamefont {Walstedt}},\ and\
  \bibinfo {author} {\bibfnamefont {R.}~\bibnamefont {Cava}},\ }\bibfield
  {title} {\bibinfo {title} {Quantum solitons in the sawtooth lattice},\
  }\href@noop {} {\bibfield  {journal} {\bibinfo  {journal} {Phys. Rev. B}\
  }\textbf {\bibinfo {volume} {53}},\ \bibinfo {pages} {6401} (\bibinfo {year}
  {1996})}\BibitemShut {NoStop}%
\bibitem [{\citenamefont {Blundell}\ and\ \citenamefont
  {N{\'u}{\~n}ez-Regueiro}(2003)}]{Blundell2003}%
  \BibitemOpen
  \bibfield  {author} {\bibinfo {author} {\bibfnamefont {S.}~\bibnamefont
  {Blundell}}\ and\ \bibinfo {author} {\bibfnamefont {M.}~\bibnamefont
  {N{\'u}{\~n}ez-Regueiro}},\ }\bibfield  {title} {\bibinfo {title} {Quantum
  topological excitations: from the sawtooth lattice to the heisenberg chain},\
  }\href@noop {} {\bibfield  {journal} {\bibinfo  {journal} {The European
  Physical Journal B-Condensed Matter and Complex Systems}\ }\textbf {\bibinfo
  {volume} {31}},\ \bibinfo {pages} {453} (\bibinfo {year} {2003})}\BibitemShut
  {NoStop}%
\bibitem [{\citenamefont {Heinze}\ \emph {et~al.}(2021)\citenamefont {Heinze},
  \citenamefont {Jeschke}, \citenamefont {Mazin}, \citenamefont
  {Metavitsiadis}, \citenamefont {Reehuis}, \citenamefont {Feyerherm},
  \citenamefont {Hoffmann}, \citenamefont {Bartkowiak}, \citenamefont
  {Prokhnenko}, \citenamefont {Wolter} \emph {et~al.}}]{Heinze2021}%
  \BibitemOpen
  \bibfield  {author} {\bibinfo {author} {\bibfnamefont {L.}~\bibnamefont
  {Heinze}}, \bibinfo {author} {\bibfnamefont {H.}~\bibnamefont {Jeschke}},
  \bibinfo {author} {\bibfnamefont {I.}~\bibnamefont {Mazin}}, \bibinfo
  {author} {\bibfnamefont {A.}~\bibnamefont {Metavitsiadis}}, \bibinfo {author}
  {\bibfnamefont {M.}~\bibnamefont {Reehuis}}, \bibinfo {author} {\bibfnamefont
  {R.}~\bibnamefont {Feyerherm}}, \bibinfo {author} {\bibfnamefont {J.-U.}\
  \bibnamefont {Hoffmann}}, \bibinfo {author} {\bibfnamefont {M.}~\bibnamefont
  {Bartkowiak}}, \bibinfo {author} {\bibfnamefont {O.}~\bibnamefont
  {Prokhnenko}}, \bibinfo {author} {\bibfnamefont {A.}~\bibnamefont {Wolter}},
  \emph {et~al.},\ }\bibfield  {title} {\bibinfo {title} {Magnetization process
  of atacamite: A case of weakly coupled s= 1/2 sawtooth chains},\ }\href@noop
  {} {\bibfield  {journal} {\bibinfo  {journal} {Physical Review Letters}\
  }\textbf {\bibinfo {volume} {126}},\ \bibinfo {pages} {207201} (\bibinfo
  {year} {2021})}\BibitemShut {NoStop}%
\bibitem [{SM()}]{SM}%
  \BibitemOpen
  \href@noop {} {}\bibinfo {note} {See Supplemental Material at [URL will be
  inserted by publisher] for the exact one-dimensional mapping, the proofs of
  sector ordering, connectivity, and center positioning, the derivation of the
  penetration depth and the field dependence of $\ell_{\rm s}$, the exact
  compact localized state, and DMRG results for two holes, which includes
  Refs.~\cite{White1992,Schollwock2011,TenPy,ogata-shiba}}\BibitemShut
  {NoStop}%
\bibitem [{\citenamefont {Mielke}(1991)}]{flat_band_2}%
  \BibitemOpen
  \bibfield  {author} {\bibinfo {author} {\bibfnamefont {A.}~\bibnamefont
  {Mielke}},\ }\bibfield  {title} {\bibinfo {title} {Ferromagnetic ground
  states for the hubbard model on line graphs},\ }\href@noop {} {\bibfield
  {journal} {\bibinfo  {journal} {Journal of Physics A: Mathematical and
  General}\ }\textbf {\bibinfo {volume} {24}},\ \bibinfo {pages} {L73}
  (\bibinfo {year} {1991})}\BibitemShut {NoStop}%
\bibitem [{\citenamefont {Tasaki}(1992)}]{flat_band_3}%
  \BibitemOpen
  \bibfield  {author} {\bibinfo {author} {\bibfnamefont {H.}~\bibnamefont
  {Tasaki}},\ }\bibfield  {title} {\bibinfo {title} {Ferromagnetism in the
  hubbard models with degenerate single-electron ground states},\ }\href@noop
  {} {\bibfield  {journal} {\bibinfo  {journal} {Physical review letters}\
  }\textbf {\bibinfo {volume} {69}},\ \bibinfo {pages} {1608} (\bibinfo {year}
  {1992})}\BibitemShut {NoStop}%
\bibitem [{\citenamefont {White}(1992)}]{White1992}%
  \BibitemOpen
  \bibfield  {author} {\bibinfo {author} {\bibfnamefont {S.~R.}\ \bibnamefont
  {White}},\ }\bibfield  {title} {\bibinfo {title} {Density matrix formulation
  for quantum renormalization groups},\ }\href@noop {} {\bibfield  {journal}
  {\bibinfo  {journal} {Physical review letters}\ }\textbf {\bibinfo {volume}
  {69}},\ \bibinfo {pages} {2863} (\bibinfo {year} {1992})}\BibitemShut
  {NoStop}%
\bibitem [{\citenamefont {Hauschild}\ \emph {et~al.}(2024)\citenamefont
  {Hauschild}, \citenamefont {Unfried}, \citenamefont {Anand}, \citenamefont
  {Andrews}, \citenamefont {Bintz}, \citenamefont {Borla}, \citenamefont
  {Divic}, \citenamefont {Drescher}, \citenamefont {Geiger}, \citenamefont
  {Hefel} \emph {et~al.}}]{TenPy}%
  \BibitemOpen
  \bibfield  {author} {\bibinfo {author} {\bibfnamefont {J.}~\bibnamefont
  {Hauschild}}, \bibinfo {author} {\bibfnamefont {J.}~\bibnamefont {Unfried}},
  \bibinfo {author} {\bibfnamefont {S.}~\bibnamefont {Anand}}, \bibinfo
  {author} {\bibfnamefont {B.}~\bibnamefont {Andrews}}, \bibinfo {author}
  {\bibfnamefont {M.}~\bibnamefont {Bintz}}, \bibinfo {author} {\bibfnamefont
  {U.}~\bibnamefont {Borla}}, \bibinfo {author} {\bibfnamefont
  {S.}~\bibnamefont {Divic}}, \bibinfo {author} {\bibfnamefont
  {M.}~\bibnamefont {Drescher}}, \bibinfo {author} {\bibfnamefont
  {J.}~\bibnamefont {Geiger}}, \bibinfo {author} {\bibfnamefont
  {M.}~\bibnamefont {Hefel}}, \emph {et~al.},\ }\bibfield  {title} {\bibinfo
  {title} {Tensor network python (tenpy) version 1},\ }\href@noop {} {\bibfield
   {journal} {\bibinfo  {journal} {SciPost Physics Codebases}\ ,\ \bibinfo
  {pages} {041}} (\bibinfo {year} {2024})}\BibitemShut {NoStop}%
\bibitem [{\citenamefont {Schollw{\"o}ck}(2011)}]{Schollwock2011}%
  \BibitemOpen
  \bibfield  {author} {\bibinfo {author} {\bibfnamefont {U.}~\bibnamefont
  {Schollw{\"o}ck}},\ }\bibfield  {title} {\bibinfo {title} {The density-matrix
  renormalization group in the age of matrix product states},\ }\href@noop {}
  {\bibfield  {journal} {\bibinfo  {journal} {Annals of physics}\ }\textbf
  {\bibinfo {volume} {326}},\ \bibinfo {pages} {96} (\bibinfo {year}
  {2011})}\BibitemShut {NoStop}%
\bibitem [{\citenamefont {Ogata}\ and\ \citenamefont
  {Shiba}(1990)}]{ogata-shiba}%
  \BibitemOpen
  \bibfield  {author} {\bibinfo {author} {\bibfnamefont {M.}~\bibnamefont
  {Ogata}}\ and\ \bibinfo {author} {\bibfnamefont {H.}~\bibnamefont {Shiba}},\
  }\bibfield  {title} {\bibinfo {title} {Bethe-ansatz wave function, momentum
  distribution, and spin correlation in the one-dimensional strongly correlated
  hubbard model},\ }\href {https://doi.org/10.1103/PhysRevB.41.2326} {\bibfield
   {journal} {\bibinfo  {journal} {Phys. Rev. B}\ }\textbf {\bibinfo {volume}
  {41}},\ \bibinfo {pages} {2326} (\bibinfo {year} {1990})}\BibitemShut
  {NoStop}%
\end{thebibliography}%

\clearpage
\onecolumngrid
\setcounter{secnumdepth}{2}
\setcounter{section}{0}
\setcounter{equation}{0}
\setcounter{figure}{0}
\setcounter{table}{0}
\setcounter{lemma}{0}
\setcounter{theorem}{0}
\setcounter{corollary}{0}
\setcounter{proposition}{0}
\setcounter{remark}{0}
\renewcommand{\thesection}{S\arabic{section}}
\renewcommand{\theequation}{S\arabic{equation}}
\renewcommand{\thefigure}{S\arabic{figure}}
\renewcommand{\thetable}{S\arabic{table}}
\renewcommand{\thelemma}{S\arabic{lemma}}
\renewcommand{\thetheorem}{S\arabic{theorem}}
\renewcommand{\theproposition}{S\arabic{proposition}}
\renewcommand{\theremark}{S\arabic{remark}}
\renewcommand{\thecorollary}{S\arabic{corollary}}

\begin{center}
{\large\bfseries Supplemental Material for ``Exact Magnetic Confinement of a Doped Hole by Kinetic Frustration''}
\end{center}

\section{Exact mapping to a one-dimensional chain}
\label{app:1D_mapping}

The exact solvability of the one-hole sawtooth Hubbard model
follows from its reduction to a simple one-dimensional chain.
We label the basal sites by $2n+1$, with $n=0,\ldots,L$, and
the apex sites by $2n$, with $n=1,\ldots,L$. For a fixed dimer
configuration $\mathcal C$, the hopping amplitude on basal bond
$n$ is $-\tau_n$, where
\begin{equation}
\tau_n=
\begin{cases}
+t_2, & \text{singlet dimer},\\
-t_2, & \text{triplet dimer}.
\end{cases}
\label{eq:tau_definition}
\end{equation}
The hopping amplitude along the sides of every triangle is
$-t_1$.

Within a fixed magnetization sector $M$, the Zeeman term $-hM$
is constant. Throughout this appendix, $H_{\mathcal C}$ denotes
the zero-field hopping Hamiltonian for configuration $\mathcal C$,
and $E$ denotes one of its kinetic eigenenergies.

Let $a_i$ be the amplitudes of an eigenstate of
$H_{\mathcal C}$ with energy $E$. The Schr\"odinger equation on
apex site $2n$ is
\begin{equation}
E a_{2n}=-t_1(a_{2n-1}+a_{2n+1}).
\label{eq:apex_amplitude}
\end{equation}
For an interior basal site $2n+1$, substitution of
Eq.~\eqref{eq:apex_amplitude} gives
\begin{align}
E a_{2n+1}
={}&-\left(\tau_{n-1}-\frac{t_1^2}{E}\right)a_{2n-1}
-\left(\tau_n-\frac{t_1^2}{E}\right)a_{2n+3}+\frac{2t_1^2}{E}a_{2n+1},
\label{eq:schr_eq_mapped}
\end{align}
with the corresponding one-sided equations at the two ends.

Introducing the basal-state vector $|a_{\rm b}\rangle$, these
equations can be written as
\begin{equation}
E|a_{\rm b}\rangle
=\mathcal K_{\mathcal C}^{\rm sgn}(E)|a_{\rm b}\rangle,
\label{eq:K_signed}
\end{equation}
where
\begin{equation}
\mathcal K_{\mathcal C}^{\rm sgn}(E)
=
H_{\rm 1D}\!\left(\{\bar t_n(E)\}\right)
+U_{\rm b}(E)
+\frac{2t_1^2}{E}\,\mathbb 1_{\rm b},
\label{eq:K_definition}
\end{equation}
and
\begin{equation}
\bar t_n(E)=\tau_n-\frac{t_1^2}{E}.
\end{equation}
Here $H_{\rm 1D}$ is an open chain on the $L+1$ basal sites
with off-diagonal matrix elements $-\bar t_n(E)$, while
\begin{equation}
U_{\rm b}(E)
=
\frac{t_1^2}{|E|}
\left(
|1\rangle\langle1|
+|2L+1\rangle\langle2L+1|
\right)
\label{eq:boundary_potential}
\end{equation}
is a positive potential on the two end sites. This term reduces
the magnitude of their diagonal contribution because each
boundary basal site has only one apex neighbor.

For $E<0$, the singlet and triplet bond magnitudes are
\begin{equation}
\tilde t_s(E)=\frac{t_1^2}{|E|}+t_2,
\qquad
\tilde t_t(E)=
\left|\frac{t_1^2}{|E|}-t_2\right|,
\label{eq:effective_hoppings}
\end{equation}
and therefore
\begin{equation}
\tilde t_s(E)>\tilde t_t(E)\geq0.
\label{eq:effective_hopping_order}
\end{equation}
The signed triplet hopping
$\bar t_t(E)=t_1^2/|E|-t_2$ can change sign and vanishes at the
destructive-interference point. Because the basal lattice is an
open chain, all nonzero hopping signs can be removed by a gauge
transformation. Only the magnitudes in
Eq.~\eqref{eq:effective_hoppings} affect the spectrum.

We denote the resulting gauge-equivalent matrix by
$\mathcal K_{\mathcal C}(E)$.

For $E\neq0$, the mapping is bijective. Restricting an eigenvector
of $H_{\mathcal C}$ to the basal sites gives an eigenvector of
$\mathcal K_{\mathcal C}(E)$ with eigenvalue $E$. Conversely, a
basal eigenvector satisfying Eq.~\eqref{eq:K_signed} uniquely
determines the apex amplitudes through
\begin{equation}
a_{2n}
=-\frac{t_1}{E}(a_{2n-1}+a_{2n+1}).
\label{eq:inverse_mapping}
\end{equation}

\subsection*{Equivalence of the ground states}

\begin{lemma}[Level counting]
\label{lem:count}
For every $E<0$, the number of eigenvalues of
$H_{\mathcal C}$ below $E$ equals the number of eigenvalues of
$\mathcal K_{\mathcal C}(E)$ below $E$.
\end{lemma}

\begin{proof}
Decompose $H_{\mathcal C}-E\mathbb 1$ into basal and apex blocks:
\begin{equation}
H_{\mathcal C}-E\mathbb 1
=
\begin{pmatrix}
A & B\\
B^{\mathsf T} & D
\end{pmatrix},
\qquad
D=-E\mathbb 1_{\rm apex}.
\end{equation}
For $E<0$, $D$ is positive definite. The congruence
transformation
\begin{align}
&
\begin{pmatrix}
\mathbb 1_{\rm b} & -BD^{-1}\\
0 & \mathbb 1_{\rm apex}
\end{pmatrix}
(H_{\mathcal C}-E\mathbb 1)
\begin{pmatrix}
\mathbb 1_{\rm b} & 0\\
-D^{-1}B^{\mathsf T} & \mathbb 1_{\rm apex}
\end{pmatrix}=
\begin{pmatrix}
A-BD^{-1}B^{\mathsf T} & 0\\
0 & D
\end{pmatrix}
\label{eq:schur_congruence}
\end{align}
gives
\begin{equation}
A-BD^{-1}B^{\mathsf T}
=
\mathcal K_{\mathcal C}^{\rm sgn}(E)
-E\mathbb 1_{\rm b}.
\end{equation}
The gauge transformation relating
$\mathcal K_{\mathcal C}^{\rm sgn}(E)$ and
$\mathcal K_{\mathcal C}(E)$ does not change their spectra.
Sylvester's law of inertia therefore implies that
$H_{\mathcal C}-E\mathbb 1$ and
$\mathcal K_{\mathcal C}(E)-E\mathbb 1_{\rm b}$ have the same
number of negative eigenvalues, proving the claim.
\end{proof}

Define
\begin{equation}
g_{\mathcal C}(E)
\equiv
\lambda_{\min}\!\left[\mathcal K_{\mathcal C}(E)\right].
\end{equation}

\begin{theorem}[Ground-state correspondence]
\label{thm:unique}
For $E<0$,
\begin{equation}
g_{\mathcal C}(E)=E
\quad\Longleftrightarrow\quad
E=E_{\mathcal C},
\end{equation}
where $E_{\mathcal C}$ is the ground-state energy of
$H_{\mathcal C}$.
\end{theorem}

\begin{proof}
Let $N_{\mathcal C}(E)$ denote the number of eigenvalues of
$H_{\mathcal C}$ below $E$. Lemma~\ref{lem:count} implies
\begin{equation}
g_{\mathcal C}(E)<E
\quad\Longleftrightarrow\quad
N_{\mathcal C}(E)\geq1.
\end{equation}
If $E<E_{\mathcal C}$, then $N_{\mathcal C}(E)=0$ and hence
$g_{\mathcal C}(E)\geq E$. Equality would make $E$ an eigenvalue
of $H_{\mathcal C}$ below its ground-state energy, so in fact
$g_{\mathcal C}(E)>E$.

If $E>E_{\mathcal C}$, the ground state lies below $E$, giving
$N_{\mathcal C}(E)\geq1$ and therefore
$g_{\mathcal C}(E)<E$.

Finally, the bijection implies that $E_{\mathcal C}$ is an
eigenvalue of $\mathcal K_{\mathcal C}(E_{\mathcal C})$, so
$g_{\mathcal C}(E_{\mathcal C})\leq E_{\mathcal C}$. Taking
$E\to E_{\mathcal C}^{-}$ and using the continuity of
$g_{\mathcal C}(E)$ gives the reverse inequality. Hence
\begin{equation}
g_{\mathcal C}(E_{\mathcal C})=E_{\mathcal C},
\end{equation}
and this is the unique negative solution of
$g_{\mathcal C}(E)=E$.
\end{proof}

The original sawtooth ground state is therefore identified
self-consistently with the lowest eigenstate of the full
energy-dependent matrix $\mathcal K_{\mathcal C}(E)$, including
its uniform diagonal term and boundary potential.

\begin{remark}[$E_{\mathcal C}<0$]
\label{rem:Eneg}
A normalized trial state with amplitudes
$(a_{2n},a_{2n+1})=(1,1)/\sqrt2$ on a single side of one triangle
has expectation value $-t_1$. Therefore
\begin{equation}
E_{\mathcal C}\leq -t_1<0,
\end{equation}
so the restriction to $E<0$ includes the true ground state for
every dimer configuration.
\end{remark}

\section{Proof of sector ordering}
\label{app:ordering}

Within a fixed magnetization sector $M$, the Zeeman term $-hM$
is constant. We therefore define $E_{\rm GS}(k)$ as the
ground-state energy of the zero-field hopping Hamiltonian among
configurations containing $k$ triplet dimers.

Let $\mathcal C$ denote a fixed arrangement of singlet and
triplet dimers, and let $E_{\mathcal C}<0$ be its ground-state
energy. As shown in Appendix~\ref{app:1D_mapping}, eliminating
the apex amplitudes at a fixed energy $E<0$ maps the sawtooth
chain onto an open chain of basal sites described by the
energy-dependent matrix $\mathcal K_{\mathcal C}(E)$. After
gauging away the hopping signs, a singlet bond carries hopping
\begin{equation}
\tilde t_s(E)=\frac{t_1^2}{|E|}+t_2,
\end{equation}
whereas a triplet bond carries
\begin{equation}
\tilde t_t(E)=
\left|\frac{t_1^2}{|E|}-t_2\right|.
\end{equation}
Thus,
\begin{equation}
\tilde t_s(E)>\tilde t_t(E)\geq0.
\label{eq:hopping_order}
\end{equation}

Consider a configuration $\mathcal C'$ obtained from
$\mathcal C$ by replacing one triplet dimer by a singlet. At
fixed $E$, this replacement increases one hopping of the mapped
chain from $\tilde t_t(E)$ to $\tilde t_s(E)$, leaving all
diagonal terms unchanged. In the case of $\tilde t_t(E)>0$, one can choose a positive normalized eigenvector $\phi$ of $\mathcal K_{\mathcal C}(E)$ as its ground state due to Perron-Frobenius theorem. The variational principle gives
\begin{align}
\lambda_{\min}[\mathcal K_{\mathcal C'}(E)]
&\leq
\langle\boldsymbol{\phi}|
\mathcal K_{\mathcal C'}(E)
|\boldsymbol{\phi}\rangle=
\lambda_{\min}[\mathcal K_{\mathcal C}(E)]
-2\bigl(\tilde t_s-\tilde t_t\bigr)
\phi_i\phi_{i+1}<
\lambda_{\min}[\mathcal K_{\mathcal C}(E)].
\label{eq:mapped_variational_order}
\end{align}
Evaluating this inequality at $E=E_{\mathcal C}$ and applying
the level-counting result once more yields
\begin{equation}
E_{\mathcal C'}<E_{\mathcal C}.
\end{equation}
Replacing triplet dimers by singlets successively therefore
gives
\begin{equation}
E_{\rm GS}(k_1)< E_{\rm GS}(k_2),
\qquad k_1< k_2.
\end{equation}
\begin{remark}[$\tilde t_t(E)=0$]
\label{rem:t_tzero}
Let $\mathcal C$ contain $k\geq1$ triplet dimers and satisfy
$\tilde t_t(E_{\mathcal C})=0$. Then
$\mathcal K_{\mathcal C}(E_{\mathcal C})$ decomposes into blocks
connected by singlet bonds. Let $\mathcal B$ be a block that
carries a ground state $\phi$ of
$\mathcal K_{\mathcal C}(E_{\mathcal C})$. By the
Perron--Frobenius theorem applied to $\mathcal B$, $\phi$ can be
chosen strictly positive on $\mathcal B$ and zero elsewhere.
Because $k\geq1$, $\mathcal B$ is not the whole chain, so at least
one triplet bond $(i,i+1)$ has exactly one endpoint, say $i$, in
$\mathcal B$. Let $\mathcal C'$ be obtained by replacing this
triplet by a singlet. Then
$\langle\phi|\mathcal K_{\mathcal C'}(E_{\mathcal C})|\phi\rangle
=\lambda_{\min}[\mathcal K_{\mathcal C}(E_{\mathcal C})]$, but
$\phi$ is not an eigenvector of
$\mathcal K_{\mathcal C'}(E_{\mathcal C})$, because the new bond
couples $\phi_i>0$ to $\phi_{i+1}=0$. Hence
$\lambda_{\min}[\mathcal K_{\mathcal C'}(E_{\mathcal C})]<E_{\mathcal C}$,
and Theorem~\ref{thm:unique} gives
$E_{\mathcal C'}<E_{\mathcal C}$. Together with
Eq.~\eqref{eq:mapped_variational_order}, this proves
$E_{\rm GS}(k_1)<E_{\rm GS}(k_2)$ for $k_1<k_2$ in all cases and
completes the proof of Theorem~\ref{thm:sector}.
\end{remark}

\section{Connectivity proof}
\label{app:connectivity}

For a fixed number $k$ of triplet dimers, we prove that the
ground-state energy is minimized when the $L-k$ singlet dimers
form a single connected region.

At a fixed physical energy $E<0$, Appendix~\ref{app:1D_mapping}
maps a dimer configuration $\mathcal C$ onto the one-dimensional
Hamiltonian
\begin{equation}
\mathcal H_{\mathcal C}(E)
=
H_{\rm 1D}\!\left(\{\tilde t_n(E)\}\right)+U_{\rm b}(E),
\label{eq:H1D_connectivity}
\end{equation}
whose eigenvalue is
\begin{equation}
E_{\rm 1D}(E)=E-\frac{2t_1^2}{E}
\label{eq:shifted_energy}
\end{equation}
and is monotonic for $E<0$. The two possible hopping magnitudes are
\begin{equation}
\tilde t_s(E)=\frac{t_1^2}{|E|}+t_2,
\qquad
\tilde t_t(E)=
\left|\frac{t_1^2}{|E|}-t_2\right|,
\end{equation}
with $\tilde t_s>\tilde t_t\geq0$. In this appendix, all effective
hoppings are evaluated at the same physical energy $E$.

\begin{theorem}[Connectivity]
\label{thm:connectivity}
At fixed $k$, the ground-state energy is minimized when all
$L-k$ singlet dimers form a single connected region.
\end{theorem}

We argue by contradiction. Since there are finitely many dimer
configurations with $k$ triplets, a globally minimizing
configuration at fixed $k$ exists. Assume that it contains at
least two singlet regions. Throughout, singlet and triplet regions
are disjoint sets of basal sites: every site incident to a singlet
bond belongs to a singlet region.
We first consider $\tilde t_t>0$, for which the mapped chain is
connected and its normalized ground-state amplitudes can be
chosen strictly positive.

The argument uses Lemmas~\ref{lem:rearr} and \ref{lem:straddle}
of the main text, which we prove first. Since
$\mathcal H_{\mathcal C}(E)$ and $\mathcal K_{\mathcal C}(E)$ differ
only by the uniform shift $2t_1^2/E$, they have the same
eigenvectors.

\begin{proof}[Proof of Lemma~\ref{lem:rearr}]
Exchange the two hopping values while keeping the amplitudes
fixed. The change in the Rayleigh quotient of the mapped
Hamiltonian is
\begin{equation}
\Delta\mathcal R
=
2(t_j-t_i)
\left(
\psi_j\psi_{j+1}-\psi_i\psi_{i+1}
\right).
\end{equation}
where $\Delta\mathcal R\equiv\mathcal R_{\rm exchanged}
-\mathcal R_{\rm original}$.
If $\psi_i\psi_{i+1}>\psi_j\psi_{j+1}$, the exchanged
configuration would have a trial eigenvalue below
$\epsilon(E)$. After restoring the uniform shift
$2t_1^2/E$, its lowest eigenvalue would lie below $E$.
Lemma~\ref{lem:count} would then imply that the corresponding
sawtooth Hamiltonian has a ground-state energy below $E$,
contradicting the assumed global minimality of the original
configuration.

If $\psi_i\psi_{i+1}=\psi_j\psi_{j+1}$, the exchanged trial state
has the same Rayleigh quotient $\epsilon(E)$. Unless the
exchanged configuration has a lower eigenvalue, which again
contradicts minimality, the trial state must be one of its
eigenvectors. Since the two bonds are distinct, bond $(i,i+1)$ has
an endpoint $x$ that is not an endpoint of $(j,j+1)$. Let $y$ be
the other endpoint of $(i,i+1)$. The eigenvalue equation at $x$
changes by $(t_i-t_j)\psi_y\neq0$, because all amplitudes are
positive. Hence the trial state is not an eigenvector, and
equality is also excluded.
\end{proof}

\begin{proof}[Proof of Lemma~\ref{lem:straddle}]
Let $p_0=j',\ldots,p_m=j$ be the path from the peak of
$\mathcal I_2$ to $j$, and let $s$ be the largest index for which
$p_s\in\mathcal I_2\cup\mathcal G$. Suppose that no bond
$(p_r,p_{r+1})$, with $r\leq s$, satisfies the claim. A
non-witnessing bond leaving a site with amplitude at least
$\psi_j$ must then enter a site with amplitude strictly greater
than $\psi_j$. Starting from
$\psi_{p_0}\geq\psi_j$ and iterating this observation gives
$\psi_{p_{s+1}}>\psi_j$. But
$p_{s+1}\in\mathcal I_1$, contradicting the definition of
$\psi_j$.
\end{proof}

\subsection*{Rearrangement step}

The five-site rearrangement is shown in
Fig.~\ref{fig:proof}(b). Let $c=j$ be the peak of
$\mathcal I_1$, let $a\in\mathcal I_1$ be its singlet neighbor,
and let $b$ be its other neighbor. The bond $(c,b)$ carries
$t_{bc}\in\{\tilde t_t,\tilde t_s\}$. Let $(e,d)$ be the bond
provided by Lemma~\ref{lem:straddle}, with hopping
$t_{ed}\in\{\tilde t_t,\tilde t_s\}$. The case $d=c$ cannot
occur. It would imply $e=b$ and $\psi_b=\psi_e\geq\psi_c$. If
$(c,b)$ is a singlet bond, then $b\in\mathcal I_1$, contradicting
$e\in\mathcal I_2\cup\mathcal G$. If $(c,b)$ is a triplet bond,
Lemma~\ref{lem:rearr} gives $\psi_c\psi_b<\psi_a\psi_c$, i.e.,
$\psi_b<\psi_a\leq\psi_c$. The two local bond
configurations are
\begin{center}
\begin{tabular}{lll}
State 1:
& $(a,c)=\tilde t_s$,
& $(c,b)=t_{bc}$,\quad $(d,e)=t_{ed}$,\\
State 2:
& $(a,b)=t_{bc}$,
& $(c,d)=t_{ed}$,\quad $(e,c)=\tilde t_s$.
\end{tabular}
\end{center}
All other bonds are unchanged. This rearrangement preserves the
numbers of singlet and triplet dimers and therefore preserves
both $k$ and the magnetization sector.

Let $\mathcal R_1=\epsilon(E)$ be the ground-state Rayleigh
quotient of State 1, and let $\mathcal R_2$ be the Rayleigh
quotient obtained by applying the rearrangement to the same
amplitudes. From Lemmas~\ref{lem:rearr} and
\ref{lem:straddle},
\begin{equation}
\psi_e\geq\psi_c\geq\psi_d,
\qquad
\psi_c\geq\psi_a,\psi_b.
\end{equation}
A direct calculation gives
\begin{align}
\frac{\mathcal R_1-\mathcal R_2}{2}
={}&
(\psi_e-\psi_c)\psi_d(\tilde t_s-t_{ed})+
(\psi_c-\psi_a)\psi_b(\tilde t_s-t_{bc})
\nonumber\\
&+
\tilde t_s(\psi_c-\psi_a)(\psi_c-\psi_b)
-
\tilde t_s(\psi_e-\psi_c)(\psi_d-\psi_c)
+\frac{\Delta U_{\rm b}}{2}
\geq0.
\label{eq:connectivity_proof}
\end{align}
Every term is non-negative because
$t_{ed},t_{bc}\leq\tilde t_s$. The boundary contribution
$\Delta U_{\rm b}$ vanishes unless $c$ is an end site of the
chain. In that case $b$ does not exist, and we set $\psi_b=0$ and
$t_{bc}=0$ in Eq.~\eqref{eq:connectivity_proof}. After the
rearrangement, $a$ becomes the end site, giving
\begin{equation}
\Delta U_{\rm b}
=
U_{\rm b}(\psi_c^2-\psi_a^2)\geq0.
\end{equation}
Thus, State 2 has a trial eigenvalue no greater than
$\epsilon(E)$.

\begin{proposition}[Strict lowering]
\label{prop:strict}
For $\tilde t_t>0$, the rearranged configuration has a strictly
lower physical ground-state energy.
\end{proposition}

\begin{proof}
If $\mathcal R_2<\epsilon(E)$, the result follows immediately
from Lemma~\ref{lem:count}. It remains to exclude equality.
Suppose that $\mathcal R_2=\epsilon(E)$ and that the rearranged
configuration has no lower eigenvalue. The rearranged trial
state must then be a ground-state eigenvector of
$\mathcal H_2(E)$.

Subtracting the eigenvalue equations of the two configurations
at the junction sites gives
\begin{equation}
t_{ed}\psi_e=t_{ed}\psi_c,
\qquad
t_{ed}\psi_d=\tilde t_s\psi_c.
\label{eq:junction_equalities}
\end{equation}
Since $t_{ed}>0$, $t_{ed}\leq\tilde t_s$, and
$\psi_e\geq\psi_c\geq\psi_d>0$, these equations require
\begin{equation}
t_{ed}=\tilde t_s,
\qquad
\psi_d=\psi_c=\psi_e.
\end{equation}
Under the equality assumption, the rearranged state is a
ground-state eigenvector of $\mathcal H_2(E)$ with eigenvalue
$E_{\rm 1D}(E)$. The eigenvalue equation at the interior site
$c$ therefore gives
\begin{equation}
E_{\rm 1D}(E)=-2\tilde t_s.
\end{equation}
However, all hoppings in $\mathcal H_2(E)$ are bounded by
$\tilde t_s$, and its boundary potential is non-negative.
Comparison with a uniform open chain of $L+1$ sites therefore
gives
\begin{align}
E_{\rm 1D}(E)
=\lambda_{\min}[\mathcal H_2(E)]
\geq
-2\tilde t_s\cos\left(\frac{\pi}{L+2}\right)
>
-2\tilde t_s,
\end{align}
which is a contradiction. Hence equality cannot occur, and the
rearrangement strictly lowers the physical ground-state energy.
\end{proof}

\textbf{Case $\tilde t_t=0$.} Let $\mathcal C$ be minimizing with
$\tilde t_t(E_{\mathcal C})=0$, and let $\mathcal B$ and $\phi$ be
defined as in Remark~\ref{rem:t_tzero}. Suppose that some singlet
bond lies outside $\mathcal B$. Take a triplet bond with exactly
one endpoint in $\mathcal B$ and replace it by a singlet. Then
replace one singlet bond outside $\mathcal B$ by a triplet, so
that $k$ is unchanged. At $E=E_{\mathcal C}$, the second change
only affects blocks on which $\phi=0$. The argument of
Remark~\ref{rem:t_tzero} therefore again gives a strictly lower
physical ground-state energy, contradicting minimality. Hence all
singlet bonds lie in $\mathcal B$ and form one connected region.


In both cases, a minimizing configuration with two or more
singlet regions can be rearranged into a configuration with the
same $k$ and a strictly lower physical ground-state energy,
contradicting minimality. Hence every minimizing configuration
has a single connected singlet region, which proves
Theorem~\ref{thm:connectivity}.

\section{Center-positioning proof}
\label{app:center_pos}

Having established connectivity, we now determine the position of
the singlet region. Let the left and right triplet regions contain
$p-1$ and $q-1$ triplet bonds, respectively. Their basal-site
amplitudes are denoted by
$\psi_{\alpha_1},\ldots,\psi_{\alpha_p}$ and
$\psi_{\beta_1},\ldots,\psi_{\beta_q}$, ordered from the chain
ends toward the singlet region. Thus, $\alpha_1$ and $\beta_1$
are the two end sites and carry the boundary potential
$U_{\rm b}>0$, while
$(\alpha_p,\alpha_{p+1})$ and
$(\beta_q,\beta_{q+1})$ are the outermost singlet bonds.
We set $\psi_{\alpha_0}=0$.

We first consider $\tilde t_s>\tilde t_t>0$. The mapped chain is
then connected, so its ground-state amplitudes can be chosen
strictly positive.

\begin{theorem}[Center positioning]
\label{thm:center_position}
For $\tilde t_s>\tilde t_t>0$, the ground-state configuration
satisfies
\begin{equation}
|p-q|\leq1.
\end{equation}
\end{theorem}

\begin{proof}
Assume, without loss of generality, that $q\geq p+2$, and suppose
that this configuration minimizes the physical ground-state
energy at fixed $k$. We construct another configuration with
tails $(p+1,q-1)$ and strictly lower energy.

Let $m$ be the smallest index in $\{1,\ldots,p+1\}$ satisfying
\begin{equation}
\psi_{\alpha_m}\geq\psi_{\beta_{m+1}}.
\label{eq:m_definition}
\end{equation}

\emph{Existence of $m$.}
Suppose that no $m\leq p$ satisfies
Eq.~\eqref{eq:m_definition}. In particular,
$\psi_{\alpha_p}<\psi_{\beta_{p+1}}$. The bond
$(\alpha_p,\alpha_{p+1})$ is a singlet bond, whereas
$(\beta_{p+1},\beta_{p+2})$ is a triplet bond because
$q\geq p+2$. Lemma~\ref{lem:rearr} therefore gives
\begin{equation}
\psi_{\alpha_p}\psi_{\alpha_{p+1}}
\geq
\psi_{\beta_{p+1}}\psi_{\beta_{p+2}}.
\end{equation}
Together with
$\psi_{\alpha_p}<\psi_{\beta_{p+1}}$, this implies
$\psi_{\alpha_{p+1}}>\psi_{\beta_{p+2}}$, so $m=p+1$
satisfies Eq.~\eqref{eq:m_definition}.

\emph{Trial state.}
Construct a trial state $\phi$ for the configuration
$(p+1,q-1)$ by exchanging the ordered segments
\begin{equation}
(\alpha_1,\ldots,\alpha_{m-1})
\quad\longleftrightarrow\quad
(\beta_1,\ldots,\beta_m).
\end{equation}
The new left and right tails are
\begin{align}
&\beta_1,\ldots,\beta_m,\alpha_m,\ldots,\alpha_p,
\nonumber\\
&\alpha_1,\ldots,\alpha_{m-1},
\beta_{m+1},\ldots,\beta_q,
\end{align}
respectively. This permutation preserves the norm. For $m\ge2$ it exchanges the two end amplitudes $\psi_{\alpha_1}$ and
$\psi_{\beta_1}$, leaving the total boundary contribution
$U_{\rm b}(\psi_{\alpha_1}^2+\psi_{\beta_1}^2)$ unchanged. For $m=1$, $\alpha_1$ becomes an interior site and $\beta_2$ becomes the right end site, so the boundary contribution changes by $U_{\rm b}(\psi_{\beta_2}^2-\psi_{\alpha_1}^2)$.

Only two bonds acquire new endpoint amplitudes. Let
\begin{equation}
\tau=
\begin{cases}
\tilde t_s, & m=p+1,\\
\tilde t_t, & m\leq p,
\end{cases}
\end{equation}
be the hopping on the bond immediately outside $\alpha_m$.
The other changed bond is a triplet bond with hopping
$\tilde t_t$.

Let $E$ be the physical ground-state energy of the assumed
minimizing configuration and let
$E_{\rm 1D}(E)=E-2t_1^2/E$ be the corresponding reduced-chain
eigenvalue. Denoting the original and trial Rayleigh quotients by
$\mathcal R_{p,q}$ and $\mathcal R_\phi$, respectively, one finds
\begin{align}
\frac{\mathcal R_{p,q}-\mathcal R_\phi}{2}
=
\bigl(\psi_{\beta_m}-\psi_{\alpha_{m-1}}\bigr)
\Big[
\tilde t_t
\bigl(\psi_{\alpha_m}-\psi_{\beta_{m+1}}\bigr)+(\tau-\tilde t_t)\psi_{\alpha_m}
\Big]+\delta_{m,1}\frac{U_b}{2}\left(\psi_{\alpha_1}^2-\psi_{\beta_2}^2\right)
\geq0.
\label{eq:center_energy_difference}
\end{align}
Here
$\psi_{\beta_m}>\psi_{\alpha_{m-1}}$ follows from the
minimality of $m$, while
$\psi_{\alpha_m}\geq\psi_{\beta_{m+1}}$ follows from its
definition. Therefore
\begin{equation}
\mathcal R_\phi
\leq
\mathcal R_{p,q}
=
E_{\rm 1D}(E).
\end{equation}
For $m=1$, the additional boundary term is non-negative because $\psi_{\alpha_1}\ge\psi_{\beta_2}$ by the definition of $m$. If the inequality is strict, the lowest eigenvalue of the
rearranged reduced Hamiltonian lies below $E_{\rm 1D}(E)$.
After restoring the uniform term $2t_1^2/E$, Lemma~\ref{lem:count}
implies that the rearranged sawtooth Hamiltonian has a physical
ground-state energy below $E$, contradicting the assumed
minimality.

It remains to exclude equality. Equality in
Eq.~\eqref{eq:center_energy_difference} can occur only when
$m\leq p$ and
$\psi_{\alpha_m}=\psi_{\beta_{m+1}}$. The trial state would
saturate the variational bound only if it were an eigenvector of
the rearranged Hamiltonian. For $m\geq2$, the site carrying
$\psi_{\alpha_m}$ changes its outer neighbor from
$\psi_{\alpha_{m-1}}$ to $\psi_{\beta_m}$. Its eigenvalue
equation could remain satisfied only if
\begin{equation}
\psi_{\alpha_{m-1}}=\psi_{\beta_m},
\end{equation}
which contradicts the minimality of $m$.

For $m=1$, the site carrying $\psi_{\beta_2}$ becomes an end
site. Equality of its original interior equation and its new
boundary equation would require
\begin{equation}
\tilde t_t\psi_{\beta_1}
+U_{\rm b}\psi_{\beta_2}=0,
\end{equation}
which is impossible because
$\tilde t_t>0$, $U_{\rm b}\geq0$, and all amplitudes are positive.
Thus, the rearranged trial state cannot be a ground-state
eigenvector, and its true ground-state energy is strictly lower.

Note that the Rayleigh-quotient comparison is performed at fixed
physical energy $E$ because the mapped hoppings depend on $E$.
Lemma~\ref{lem:count} is then required to convert the lowering of
the reduced-chain eigenvalue into a lowering of the physical
sawtooth-chain energy.

We have therefore contradicted the assumption $q\geq p+2$.
The case $p\geq q+2$ follows by reflection, proving
$|p-q|\leq1$.
\end{proof}

Since the left and right regions contain $p-1$ and $q-1$ triplet
bonds, respectively, an even number of triplet bonds is divided
equally between the two ends, while an odd number differs by one.

\begin{remark}[Compact point]
When $\tilde t_t=0$, the preceding strict-positioning argument
does not apply. At
$t_2/t_1=t_{\ell_{\rm s}}$, Appendix~\ref{app:cls_exact}
constructs an exact compact localized state whose energy is
independent of the position of the singlet region when the singlet island is separated by at least one triplet dimer from each end of the chain.  All allowed positions are therefore degenerate. End-touching configurations have a higher energy due to boundary potential $U_{\rm b}$ and do not belong to this compact-state manifold.
\end{remark}

\section{Penetration depth, compact localized states, and thermodynamic limit solution}
\label{app:thermo_limit}

In this appendix, we derive the penetration depth of the hole into the
triplet regions and identify the couplings at which the hole wavefunction
becomes strictly compact. We consider an open sawtooth chain with $L+1$
basal sites, labeled $i=1,\ldots,L+1$, and $L$ apex sites, labeled
$j=1,\ldots,L$. The amplitudes on the basal and apex sites are denoted by
$\psi_i$ and $\phi_j$, respectively. A basal bond $i$ is assigned
$\tau_i=+t_2$ when its dimer is a singlet and $\tau_i=-t_2$ when its
dimer is a triplet, so that its hopping matrix element is $-\tau_i$.

\subsection{Exact recurrence and kinetic interference}
\label{app:Teff}

For an interior basal site, the Schr\"odinger equations are
\begin{align}
E\psi_i
&=-\tau_{i-1}\psi_{i-1}-\tau_i\psi_{i+1}
-t_1(\phi_{i-1}+\phi_i),
\label{eq:basal_recurrence}
\\
E\phi_j&=-t_1(\psi_j+\psi_{j+1}),
\label{eq:apex_recurrence}
\end{align}
with the absent terms omitted at the physical boundaries. Eliminating the
apex amplitudes gives
\begin{equation}
E_{\rm 1D}(E)\psi_i
=T_{i-1}(E)\psi_{i-1}+T_i(E)\psi_{i+1},
\label{eq:effective_recurrence}
\end{equation}
where
\begin{equation}
E_{\rm 1D}(E)=E-\frac{2t_1^2}{E},
\qquad
T_i(E)=-\tau_i+\frac{t_1^2}{E}.
\label{eq:effective_parameters}
\end{equation}
The one-sided boundary equations are equivalently represented by the
non-negative boundary potential introduced in
Appendix~\ref{app:1D_mapping}.

For singlet and triplet bonds, respectively,
\begin{equation}
T_{\rm s}(E)=-t_2+\frac{t_1^2}{E},
\qquad
T_{\rm t}(E)=+t_2+\frac{t_1^2}{E}.
\label{eq:singlet_triplet_hoppings}
\end{equation}
For $E<0$, their gauge-equivalent magnitudes are
\begin{equation}
\widetilde t_{\rm s}(E)=\frac{t_1^2}{|E|}+t_2,
\qquad
\widetilde t_{\rm t}(E)
=\left|\frac{t_1^2}{|E|}-t_2\right|.
\label{eq:effective_hopping_magnitudes}
\end{equation}
The triplet hopping vanishes at
\begin{equation}
E_\ast=-\frac{t_1^2}{t_2},
\label{eq:E_star}
\end{equation}
whereas $T_{\rm s}(E_\ast)=-2t_2$. Thus, the direct basal and
apex-mediated paths interfere destructively across a triplet dimer and
constructively across a singlet dimer.

\subsection{Ground state in thermodynamic limit}

After gauging away the hopping
signs (doing so preserve the overall wavefunction normalization condition), let $q$ be the wavevector inside the singlet region and
$\kappa=L_{\rm loc}^{-1}$ the inverse penetration depth in the triplet
region. The even ground-state envelope centered at $x=0$ has the form
\begin{equation}
\psi(x)=
\begin{cases}
B\cos(qx),& |x|<\ell_{\rm s}/2,\\[2pt]
A\exp[-\kappa(|x|-\ell_{\rm s}/2)],
&|x|>\ell_{\rm s}/2.
\end{cases}
\label{eq:continuum_profile}
\end{equation}
Continuity of the wavefunction gives
\begin{equation}
A=B\cos\!\left(\frac{q\ell_{\rm s}}{2}\right).
\label{eq:continuum_continuity}
\end{equation}
Because the kinetic coefficients differ on the two sides, the second
matching condition is continuity of the current, not of the derivative. This can be derived from the exact Schrödinger equation at singlet triplet boundary:
\begin{align}
    \Tilde t_t(\psi_{\ell/2+1}-\psi_{\ell/2})-\Tilde t_s(\psi_{\ell/2}-\psi_{\ell/2-1})=-(E_{1D} + \Tilde t_t + \Tilde t_s)\psi_{\ell/2},
    \label{eq:boundary_current}
\end{align}
so that, expanding $\psi_{\ell/2\pm1}-\psi_{\ell/2}$ in continuum limit ($q,\kappa\ll1$) would give
\begin{equation}
\widetilde t_{\rm s}\,q
\tan\!\left(\frac{q\ell_{\rm s}}{2}\right)
=\widetilde t_{\rm t}\,\kappa,
\label{eq:continuum_matching}
\end{equation}
where the effective hoppings are evaluated at the eigenenergy. Together
with the bulk dispersion relations,
\begin{align}
E_{1D}=-2\Tilde{t}_s\cos q,\quad E_{1D}=-2\Tilde{t}_t\cosh \kappa,
\label{eq:bulk_dispersions}
\end{align}
Eq.~\eqref{eq:continuum_matching}
determines $q$ and $\kappa$. The asymptotic form of $q$ follows most transparently from the exact
lattice version of the matching condition. Inserting
$\psi_{\ell_{\rm s}/2+n}=\psi_{\ell_{\rm s}/2}e^{-\kappa n}$ ($n\geq0$)
and $\psi_{\ell_{\rm s}/2+n}=B\cos[q(\ell_{\rm s}/2+n)]$ ($n\leq0$) into
Eq.~\eqref{eq:boundary_current} and using
Eq.~\eqref{eq:bulk_dispersions} gives, without further approximation,
\begin{equation}
\widetilde t_{\rm s}\sin q\,\tan\!\left(\frac{q\ell_{\rm s}}{2}\right)
=\widetilde t_{\rm t}\sinh\kappa
=\sqrt{\widetilde t_{\rm s}^{\,2}\cos^2q-\widetilde t_{\rm t}^{\,2}},
\label{eq:lattice_matching}
\end{equation}
where the second equality uses
$\widetilde t_{\rm t}\cosh\kappa=\widetilde t_{\rm s}\cos q$.
Equation~\eqref{eq:continuum_matching} is its $q,\kappa\ll1$ limit.
The ground state is even and nodeless, so
$0<q\ell_{\rm s}/2<\pi/2$, i.e., $q<\pi/\ell_{\rm s}\to0$. The
right-hand side of Eq.~\eqref{eq:lattice_matching} then tends to
$\sqrt{\widetilde t_{\rm s}^{\,2}-\widetilde t_{\rm t}^{\,2}}$, which is
strictly positive because $\widetilde t_{\rm s}>\widetilde t_{\rm t}\geq0$
[Eq.~\eqref{eq:effective_hopping_order}]. Since $\sin q\to0$ on the
left-hand side, $\tan(q\ell_{\rm s}/2)$ must diverge. Writing
$q\ell_{\rm s}/2=\pi/2-\delta$ gives
$\delta\simeq q\,\widetilde t_{\rm s}/
\sqrt{\widetilde t_{\rm s}^{\,2}-\widetilde t_{\rm t}^{\,2}}
=q\coth\kappa_\infty$, where $\cosh\kappa_\infty=\widetilde t_{\rm s}/
\widetilde t_{\rm t}$ is evaluated at $E_\infty$. Hence
\begin{equation}
q=\frac{\pi}{\ell_{\rm s}+2\coth\kappa_\infty}
+O(\ell_{\rm s}^{-3})
=\frac{\pi}{\ell_{\rm s}}+O(\ell_{\rm s}^{-2}).
\label{eq:q_asymptotic}
\end{equation}
The hole thus behaves as a particle in a box whose effective length
exceeds $\ell_{\rm s}$ by the extrapolation length $\coth\kappa_\infty$
at each end, reflecting its penetration into the triplet regions. At
the compact point, $\widetilde t_{\rm t}\to0$ and
$\coth\kappa_\infty\to1$, and Eq.~\eqref{eq:q_asymptotic} reproduces
the exact result $q=\pi/(\ell_{\rm s}+2)$ of
Appendix~\ref{app:cls_exact}. In the continuum regime $\kappa\ll1$,
$\coth\kappa_\infty\simeq\widetilde t_{\rm s}/(\widetilde t_{\rm t}\kappa)$,
consistent with Eq.~\eqref{eq:continuum_matching}. Only the leading
term $\pi/\ell_{\rm s}$ enters the coefficient $C$ below; the
extrapolation length first affects the $O(\ell_{\rm s}^{-3})$ term of
the energy.

Now, recalling the definitions of $E_{1D},\,\Tilde{t}_s,\,\Tilde{t}_t$ would then provide for an asymptotic result for $\varepsilon(\ell_s)\equiv |E|/t_1$:
\begin{align}
\varepsilon(\ell_s)=\varepsilon_\infty -\frac{C}{\ell_s^2}+\mathcal{O}(\ell_s^{-3}),\quad C=\frac{\pi^2}{2}\left(t+\frac{t^2+1}{\sqrt{t^2+4}}\right),\nonumber\\\varepsilon_\infty=t+\sqrt{t^2+4},\quad t\equiv t_2/t_1.
\label{eq:finite_size_bound}
\end{align}
Let us now estimate the weight of the ground state wavefunction tails in the continuum limit. Equation~\eqref{eq:continuum_matching} then gives
\begin{equation}
A\simeq
B\frac{\widetilde t_{\rm s}q}
{\widetilde t_{\rm t}\kappa}.
\end{equation}
and $B\simeq\sqrt{2/\ell_{\rm s}}$ for $\ell_s\gg1$. 
The probability contained in either tail consequently scales as
\begin{equation}
\int_{\ell_{\rm s}/2}^{\infty}|\psi(x)|^2dx
\simeq
\frac{\pi^2\widetilde t_{\rm s}^{\,2}}
{\widetilde t_{\rm t}^{\,2}\ell_{\rm s}^{\,3}\kappa^3}.
\label{eq:continuum_tail_weight}
\end{equation}
This continuum construction becomes singular as
$\widetilde t_{\rm t}\to0$ and must not be used at the compact point;
there the exact lattice solution of
Eqs.~\eqref{eq:cls_vector}--\eqref{eq:tells_SM} applies.

\subsection{Confinement length $l_s(h)$ behavior with magnetic field}    
In a magnetic field $h$, the energy of this sector is
\begin{equation}
\mathcal E_{\ell_{\rm s}}(h)
=E_{\ell_{\rm s}}-h(L-\ell_{\rm s})
=-hL+E_{\ell_{\rm s}}+h\ell_{\rm s}.
\label{eq:E_field_SM}
\end{equation}
For a given $\ell_{\rm s}$ to be the field-dependent ground state, it
must compete with every other singlet length $r$. Its exact stability
interval is therefore
\begin{equation}
\max_{r>\ell_{\rm s}}
\frac{E_{\ell_{\rm s}}-E_r}{r-\ell_{\rm s}}
<h<
\min_{r<\ell_{\rm s}}
\frac{E_r-E_{\ell_{\rm s}}}{\ell_{\rm s}-r}.
\label{eq:field_plateau_SM}
\end{equation}
The bounded singlet-region length at fixed $h>0$
follows directly: after dropping the constant $-hL$,
the quantity to be minimized is
$E_{\ell_{\rm s}}+h\ell_{\rm s}$. Since
$E_{\ell_{\rm s}}\geq E_\infty$, this expression diverges
as $\ell_{\rm s}\to\infty$, and hence its minimizing
length remains finite as $L\to\infty$. This implies
magnetic confinement when the minimizing region is
nonempty, as occurs below saturation; above saturation,
the minimum is attained at $\ell_{\rm s}=0$ and there
is no singlet cage.

The field-selected length is also a non-increasing function of $h$.
Let $\ell_1$ and $\ell_2$ be minimizers of $E_\ell+h\ell$ at fields
$h_1<h_2$. Then $E_{\ell_1}+h_1\ell_1\le E_{\ell_2}+h_1\ell_2$ and
$E_{\ell_2}+h_2\ell_2\le E_{\ell_1}+h_2\ell_1$. Adding these
inequalities gives $(h_2-h_1)(\ell_1-\ell_2)\ge0$, i.e.,
$\ell_2\le\ell_1$.

In the limit $\ell_s\gg1$ the $l_s(h)$ can be found analytically. Minimizing $\mathcal E_{\ell_s}(h)$ of Eq.~\eqref{eq:E_field_SM} with $E_{\ell_s}$ from Eq.~\eqref{eq:finite_size_bound} would give $\ell_s(h)=(2t_1C/h)^{1/3}$, plotted against the exact $\ell_s(h)$ at Fig.~\ref{fig:confinement}.

\subsection{Lattice penetration depth}
\label{app:lattice_Lloc}

In a uniform triplet region, Eq.~\eqref{eq:effective_recurrence} becomes
\begin{equation}
E_{\rm 1D}(E)\psi_i
=T_{\rm t}(E)(\psi_{i-1}+\psi_{i+1}).
\end{equation}
Writing the decaying solution as $\psi_i\propto z^i$, with $|z|<1$, gives
\begin{equation}
z+z^{-1}=\frac{E_{\rm 1D}(E)}{T_{\rm t}(E)}.
\end{equation}
The exact inverse penetration depth at a finite-size energy $E$ is therefore
\begin{equation}
L_{\rm loc}^{-1}(E)
=\operatorname{arccosh}\!\left(
\frac{|E_{\rm 1D}(E)|}{2|T_{\rm t}(E)|}
\right).
\label{eq:Lloc_finite_energy}
\end{equation}
When $E_{\rm 1D}(E)/T_{\rm t}(E)<-2$, the decaying solution also carries
the staggered factor $(-1)^i$.

In the thermodynamic limit of a long singlet chain $\ell_s\gg1$, plugging the result of $\varepsilon=\varepsilon_\infty$ from Eq.~\eqref{eq:finite_size_bound} into Eq.~\eqref{eq:Lloc_finite_energy} gives
\begin{equation}
L_{\rm loc}^{-1}(t)
=\operatorname{arccosh}\!\left[
\frac{\varepsilon_\infty(\varepsilon_\infty+2t)}
{4|t\varepsilon_\infty-1|}
\right].
\label{eq:Lloc_exact_SM}
\end{equation}
This quantity is the penetration depth into the triplet region; it is not the total spatial extent of the state, which is controlled primarily by $\ell_{\rm s}$.

Expanding Eq.~\eqref{eq:Lloc_exact_SM} for $L_{\rm loc}\gg1$ gives
\begin{equation}
L_{\rm loc}(t)\simeq
\begin{cases}
\displaystyle
\sqrt{\frac{1-t\varepsilon_\infty}{4t\varepsilon_\infty}},
&t<1/\sqrt6,\\[8pt]
\displaystyle
\sqrt{\frac{t\varepsilon_\infty-1}{4}},
&t>1/\sqrt6.
\end{cases}
\label{eq:Lloc_continuum_SM}
\end{equation}
Equation~\eqref{eq:Lloc_continuum_SM} is only a long-length expansion of
the exact lattice expression and should not be used near the compact point.
Both formulas give $L_{\rm loc}=0$ at $t=1/\sqrt6$ in the limit
$\ell_{\rm s}\to\infty$. Crossing this point changes the sign of the
effective triplet hopping and hence introduces the staggered factor in the
decaying tail.

\section{Exact compact state at finite $\ell_{\rm s}$}
\label{app:cls_exact}

We next consider a singlet region consisting of the consecutive basal
bonds
\begin{equation}
a,a+1,\ldots,a+\ell_{\rm s}-1.
\end{equation}
For the construction below, the island is assumed to be separated from
each physical end by at least one triplet bond, so that
$2\leq a\leq L-\ell_{\rm s}$. Define
\begin{equation}
\theta_{\ell_{\rm s}}=\frac{\pi}{\ell_{\rm s}+2}.
\end{equation}
The basal amplitudes
\begin{equation}
\psi_{a+n}=\mathcal N
\sin[(n+1)\theta_{\ell_{\rm s}}],
\qquad n=0,\ldots,\ell_{\rm s},
\label{eq:cls_vector}
\end{equation}
with
\begin{equation}
\psi_{a-1}=\psi_{a+\ell_{\rm s}+1}=0
\end{equation}
and all remaining basal amplitudes zero form an exact eigenstate when
\begin{equation}
E_\ast=-\frac{t_1^2}{t_2}
=-2t_2\left[1+2\cos\theta_{\ell_{\rm s}}\right].
\label{eq:cls_consistency}
\end{equation}
The corresponding apex amplitudes are fixed by
\begin{equation}
\phi_j=-\frac{t_1}{E_\ast}(\psi_j+\psi_{j+1}).
\label{eq:cls_apex}
\end{equation}
Equation~\eqref{eq:cls_consistency} is equivalent to
\begin{equation}
t_{\ell_{\rm s}}
\equiv\frac{t_2}{t_1}
=\left[
2+4\cos\!\left(\frac{\pi}{\ell_{\rm s}+2}\right)
\right]^{-1/2}
\xrightarrow[\ell_{\rm s}\to\infty]{}\frac{1}{\sqrt6}.
\label{eq:tells_SM}
\end{equation}

To verify the construction, first note that
$T_{\rm t}(E_\ast)=0$. The node equations at
$a-1$ and $a+\ell_{\rm s}+1$ therefore receive no contribution from the
neighboring nonzero basal amplitude. Inside the island,
Eq.~\eqref{eq:effective_recurrence} becomes a uniform recurrence with
$T_{\rm s}(E_\ast)=-2t_2$, and the sine profile
Eq.~\eqref{eq:cls_vector} solves it precisely when
Eq.~\eqref{eq:cls_consistency} holds. The state consequently has strictly
zero amplitude outside the singlet region, its $\ell_{\rm s}$ apex
sites, and the apex sites of the two neighboring triplet triangles.

The construction requires one triplet bond on each side and therefore
applies when $L-\ell_{\rm s}\geq2$. An island touching a physical end does
not satisfy the same node equation and is not part of this degenerate
compact-state manifold, as shown below.

\subsection*{The compact state is a ground state}

Let $\mathcal C$ be any dimer configuration containing the interior
singlet region above, and let $\mathcal K_{\mathcal C}(E)$ denote the
energy-dependent basal matrix defined by the exact mapping in
Appendix~\ref{app:1D_mapping}. At $E=E_\ast$, all effective triplet
hoppings vanish, and $\mathcal K_{\mathcal C}(E_\ast)-E_\ast\mathbb1$
decomposes into independent blocks.

The block associated with the singlet region has diagonal
$4t_2\cos\theta_{\ell_{\rm s}}$ and off-diagonal elements $-2t_2$ after
a gauge transformation. Its eigenvalues are
\begin{equation}
\lambda_m
=4t_2\left[
\cos\theta_{\ell_{\rm s}}
-\cos\!\left(\frac{m\pi}{\ell_{\rm s}+2}\right)
\right],
\qquad m=1,\ldots,\ell_{\rm s}+1,
\label{eq:cls_block_spectrum}
\end{equation}
and are non-negative, with a single zero eigenvalue at $m=1$.
Every interior basal site outside this block is isolated and has diagonal
element
\begin{equation}
\frac{2t_1^2}{E_\ast}-E_\ast
=4t_2\cos\theta_{\ell_{\rm s}}>0,
\end{equation}
whereas either physical end site has diagonal element
\begin{equation}
\frac{t_1^2}{E_\ast}-E_\ast
=t_2\left[1+4\cos\theta_{\ell_{\rm s}}\right]>0.
\end{equation}
Hence
\begin{equation}
\mathcal K_{\mathcal C}(E_\ast)-E_\ast\mathbb1\succeq0,
\end{equation}
and its smallest eigenvalue is exactly zero. Equivalently,
\begin{equation}
\lambda_{\min}[\mathcal K_{\mathcal C}(E_\ast)]=E_\ast.
\end{equation}
By the level-counting theorem of Appendix~\ref{app:1D_mapping}, $E_\ast$
is the physical ground-state energy for every interior placement.

For comparison, if the singlet region touches a physical end, its block
in $\mathcal K_{\mathcal C}(E_\ast)-E_\ast\mathbb1$ differs from the
interior block by the positive rank-one term
$t_2|e\rangle\langle e|$ on the physical end site. The interior block is
positive semidefinite with a one-dimensional kernel, and its zero-mode
wavefunction has nonzero amplitude on $|e\rangle$. The boundary block is
therefore positive definite, so that
\begin{equation}
\lambda_{\min}[\mathcal K_{\mathcal C}(E_\ast)]>E_\ast.
\end{equation}
The ground-state correspondence of Appendix~\ref{app:1D_mapping} then
implies that the end-touching configuration has physical ground-state
energy strictly greater than $E_\ast$. Thus, the compact ground states
occur precisely at the interior positions
$2\leq a\leq L-\ell_{\rm s}$.

For fixed $\ell_{\rm s}$, the construction supplies
$L-\ell_{\rm s}-1$ interior position-degenerate compact ground states
whenever $L-\ell_{\rm s}\geq2$. Their number therefore grows linearly
with $L$ in the field-selected regime $\ell_{\rm s}=O(1)$.


\section{Two holes: one cage or two?}
\label{app:two_hole}

For one hole, Theorem~\ref{thm:exact_gs} shows that the singlet dimers
of the ground state always form a single connected region, the cage of
length $\ell_{\rm s}$, and that a field fixes $\ell_{\rm s}$
independently of $L$. The natural two-hole question is whether this
picture survives: do the two holes share one field-selected cage, or
does each hole build a cage of its own? The exact solution does not
extend to two holes since the residual graph no longer has a unique
dimer covering. We therefore answer this question numerically with the
density-matrix renormalization group
(DMRG)~\cite{White1992,Schollwock2011}, and use the one-hole results
both as exact benchmarks and to interpret the outcome.

For the hopping ratios studied,
$t_2/t_1=0.1,\,0.4,$ and $0.7$, the numerical results
support a shared singlet cage at sufficiently small
positive fields. Its size is selected by the field,
and the finite-size behavior supports a finite
hole separation as $L\to\infty$. At $t_2/t_1=0.1$,
two separate cages are favored in a narrow field
window just below saturation, where the penetration
depth $L_{\rm loc}$ is comparable to the cage length.
No such window is found at $t_2/t_1=0.4$ or $0.7$.
At $h=0$, the separation distribution instead
broadens with system size, consistent with
unbound holes.

\subsection{Shared and separate cages}

With two holes, the chain holds $2L-1$ electrons, and
$S^z_{\rm tot}=(2L-1)/2-N_\downarrow$, where $N_\downarrow$ is the
number of spins antiparallel to the applied magnetic field. Since each singlet dimer carries one antiparallel
spin, $N_\downarrow$ plays the role of the one-hole singlet length
$\ell_{\rm s}=L-M$. Let $E^{(2)}(N_\downarrow)$ denote the zero-field
two-hole ground-state energy at fixed $N_\downarrow$. As in
Eq.~\eqref{eq:field_control}, a field selects
\begin{equation}
N_\downarrow^*(h)=\operatorname*{argmin}_{N_\downarrow}
\bigl[E^{(2)}(N_\downarrow)+hN_\downarrow\bigr].
\label{eq:Ndown_field}
\end{equation}
Once the singlet regions fit in the chain, $N_\downarrow^*$ does not
depend on $L$. Following a fixed sector $N_\downarrow^*$ as $L$ grows
is therefore a scan at fixed field. At $h=0$ we use
$N_\downarrow=L-1$ ($S^z_{\rm tot}=1/2$), which contains the ground
state by SU(2) symmetry.

At fixed $N_\downarrow$, the minority spins can be arranged in two
competing ways.

\textit{Separate cages.} The singlet dimers form two regions of
lengths $\ell_1+\ell_2=N_\downarrow$, separated by triplet dimers,
with one hole in each. Each region is then a one-hole cage of
the main text. Up to corrections exponentially small in their
separation, the energy is $E_{\ell_1}+E_{\ell_2}$, which is minimized
by $\ell_1=\ell_2=N_\downarrow/2$ for even $N_\downarrow$ and by
$(\ell_1,\ell_2)=(k,k+1)$ for $N_\downarrow=2k+1$. 
Thus, the energy becomes independent of the distance
between the cages at large separation, consistent
with two unbound polarons in the thermodynamic limit.

\textit{Shared cage.} All $N_\downarrow$ singlet dimers form one
region that holds both holes, together with the one electron that
cannot be paired into a dimer. Both holes then reside in a region of
fixed size, and their separation is bounded by the cage length for
every $L$.

The one-body hole density (see Fig.~\ref{fig:2holes}) cannot distinguish these arrangements. For
example, a pair in a shared cage with a delocalized center of mass has
the same smooth density as two independent holes. We therefore use the
distribution of the relative separation,
\begin{equation}
P(d)=\sum_i\langle n^h_i\,n^h_{i+d}\rangle,\qquad
\sum_{d\geq1}P(d)=1,
\label{eq:Pd}
\end{equation}
where $n^h_i=1-n_i$ is the hole number on site $i$. In a shared cage of
fixed size, the standard ${\rm std}(d)$ of $P(d)$ does not depend on $L$.
If the holes occupy separate cages, or share a cage that spans the
chain, the width grows linearly with $N_{\rm s}=2L+1$.

\subsection{Numerical method}

We used the TeNPy implementation of DMRG~\cite{TenPy} with conserved
$N$ and $S^z_{\rm tot}$: two-site DMRG with a density-matrix mixer,
followed by single-site sweeps, with bond dimension $\chi\leq600$ and
singular values below $10^{-8}$ discarded. Converged states have maximum
energy variance $\lesssim9\times 10^{-12}\,t_1^2$. For example, the
zero-field state of Fig.~\ref{fig:2holes} ($L=100$) converged in 28
sweeps at $\chi=152$, with discarded weight below $4\times10^{-15}$.
$P(d)$ follows from
$\langle n^h_in^h_j\rangle
=1-\langle n_i\rangle-\langle n_j\rangle+\langle n_in_j\rangle$.

We simulated both magnetic field and zero magnetic field cases in chains up to $L=128$ triangles. Shared and separate cages are nearly degenerate, and DMRG can remain in the arrangement it starts from. In the field scan, every sector was
therefore started from two product states, one with both holes in one
singlet region and one with each hole in its own region, and the lower
energy was kept. This comparison also yields the energy difference
between the two arrangements quoted below.

\begin{figure*}[t]
\centering
\includegraphics[width=0.9\linewidth]{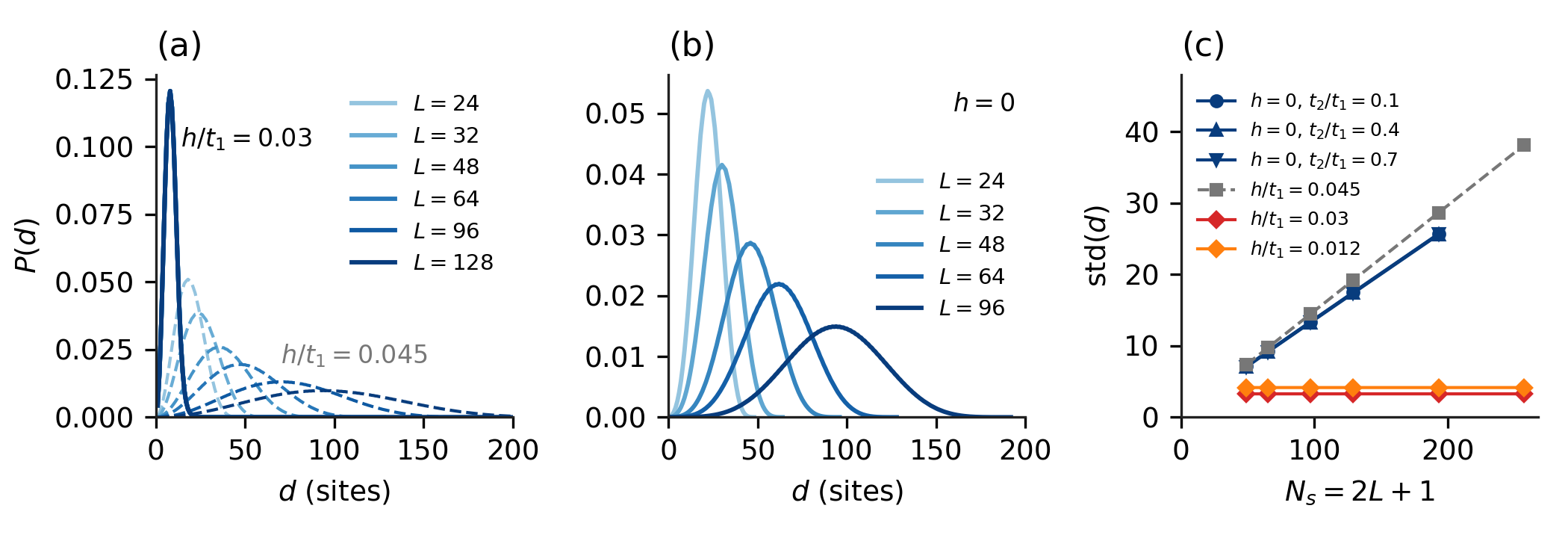}
\caption{Two holes, DMRG results.
(a) Relative-separation distribution $P(d)$,
Eq.~\eqref{eq:Pd}, at $t_2/t_1=0.1$. Solid: $h/t_1=0.03$, where
both holes share one singlet cage ($N_\downarrow=6$); the curves
for $L=24$--$128$ coincide. Dashed: $h/t_1=0.045$, above
saturation; the distribution spreads over the chain.
(b) $P(d)$ at $h=0$, $t_2/t_1=0.1$: the peak moves to
$d\simeq N_{\rm s}/2$ as $L$ grows.
(c) Standard deviation of the hole separation ${\rm std}(d)$ versus $N_{\rm s}=2L+1$, taken from calculated $P(d)$. At $h=0$
(blue, three couplings) and above saturation (grey) it grows
linearly with $N_{\rm s}$; in the field-selected shared-cage sectors
(red, orange) it is independent of $N_{\rm s}$.}
\label{fig:two_hole}
\end{figure*}

\subsection{Finite field}

At $t_2/t_1=0.1$ two fields divide the field axis into three
regimes [Fig.~\ref{fig:two_hole_field}(a)]. The lower one, $h_b$, is $L$-independent to $10^{-6}t_1$ for
$L\geq48$ and to $5\times10^{-5}t_1$ at $L=24$. The saturation field
$h_{\rm sat}$ approaches its $L\to\infty$ value from above, from
$0.0431\,t_1$ at $L=24$ to $0.0393\,t_1$ at $L=128$, because the
energy of the free holes above saturation has a finite-size
correction.

\begin{figure*}[t]
\centering
\includegraphics[width=0.9\linewidth]{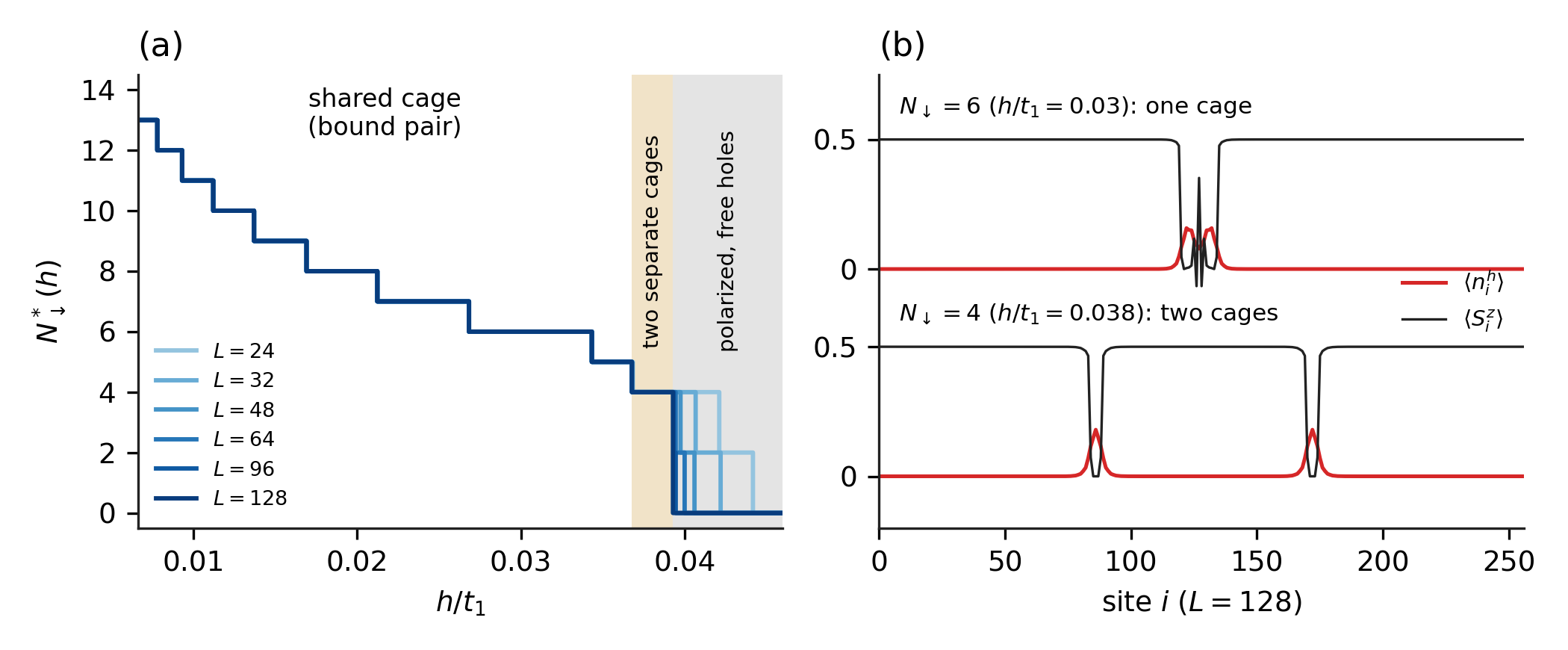}
\caption{Two holes in a field, $t_2/t_1=0.1$.
(a) Field-selected minority-spin number $N_\downarrow^*(h)$ for
$L=24$--$128$. The curves coincide except in a finite-size strip
near saturation, which is widest at $L=24$. Shaded: two separate
cages ($h_b<h<h_{\rm sat}$) and the polarized state
($h>h_{\rm sat}$).
(b) Hole density $\langle n^h_i\rangle$ and $\langle S^z_i\rangle$ at
$L=128$: one shared cage ($N_\downarrow=6$) and two separate cages
($N_\downarrow=4$) with a fully polarized region between them.}
\label{fig:two_hole_field}
\end{figure*}

(i) Above $h_{\rm sat}\simeq0.0392\,t_1$ the background is polarized
and the holes are free spinless fermions. $P(d)$ spreads over the
chain [dashed curves in Fig.~\ref{fig:two_hole}(a)], and its width
grows linearly with $N_{\rm s}$ [Fig.~\ref{fig:two_hole}(c)].

(ii) Between $h_b\simeq0.0368\,t_1$ and $h_{\rm sat}$,
$N_\downarrow^*=4$ and each hole occupies its own $\ell_{\rm s}=2$
singlet region [Fig.~\ref{fig:two_hole_field}(b)].  
The energy of two separated cages approaches
$2E_1(2)$ exponentially in their separation $d$, consistent with two unbound polarons in the
thermodynamic limit.
The saturation field coincides with the one-hole value.

(iii) Below $h_b$ one singlet region supporting both holes is lower than
two separate regions. 
$P(d)$ is independent of $L$ for $24\leq L\leq128$
[solid curves in Fig.~\ref{fig:two_hole}(a)], supporting a bound pair in the thermodynamic limit.
Its size follows the field-selected region and shrinks
in steps with increasing field. The shared singlet region
also holds the one electron that cannot be paired into a dimer; its
net spin $\langle S^z\rangle\simeq1/2$ sits between the two holes
[Fig.~\ref{fig:two_hole_field}(b)].

At $t_2/t_1=0.4$ and $0.7$ one shared singlet region is lower in energy than
two separate ones in every field-selected sector below saturation. 
There is no window of separate regions, and
the pair forms as soon as the field drops below saturation.

\subsection{Zero field}

At $h=0$, the finite-size results are consistent
with unbound holes for all three hopping ratios
studied (see Fig.~\ref{fig:two_hole}(b)).
$P(d)$ peaks at
$\langle d\rangle\simeq N_{\rm s}/2$, and 
the width grows linearly with the system size,
${\rm std}(d)\simeq0.13\,N_{\rm s}+0.8$ for all three couplings (fit to
$L=24$--$96$ from Fig.~\ref{fig:two_hole}(c)).

At $t_2=0$ the sawtooth chain is an open nearest-neighbor chain.
Because particles cannot pass one another at $U=\infty$, the spin
sequence does not affect the energy and the holes are free spinless
fermions~\cite{ogata-shiba}. In the ground state they occupy the
orbitals $\sin(qj)$ and $\sin(2qj)$, $q=\pi/(N_{\rm s}+1)$, and the one-body density is
\begin{equation}
n_h(x)
=\frac{2}{N_{\rm s}+1}
\left[\sin^2(qx)+\sin^2(2qx)\right],
\qquad x=1,\ldots,N_{\rm s}.
\label{eq:twoorbital}
\end{equation}
This density is already bimodal, with maxima at
$x/(N_{\rm s}+1)=(2\pi)^{-1}\arccos(-1/4)\simeq0.29$ and $0.71$, and a
central minimum equal to $16/25\simeq0.64$ of the maximum. Using the exact result of Eq.~\eqref{eq:twoorbital}, one can show that the respective average separation $\langle d\rangle$ is then $d\simeq0.36\,N_{\rm s}$, with
${\rm std}(d)\simeq0.15\,N_{\rm s}$ in the thermodynamic limit. By comparison to our results of Fig.~\ref{fig:two_hole}, one can conclude that the holes are therefore further apart and their separation is narrower
than for free fermions, i.e.\ they avoid each other more strongly. The
one-body density shows the same trend: its maxima are at
$0.25\,N_{\rm s}$ and $0.75\,N_{\rm s}$ instead of $0.29$ and $0.71$, and the
central value is a few percent of the maximum instead of $0.64$. The width grows with $N_{\rm s}$ because at $h=0$ the singlet region
spans the chain, so nothing confines the relative coordinate. Only when
a field makes the singlet region finite can the two holes share one
region.
\begin{figure}[t]
    \centering
    \includegraphics[width=\linewidth]{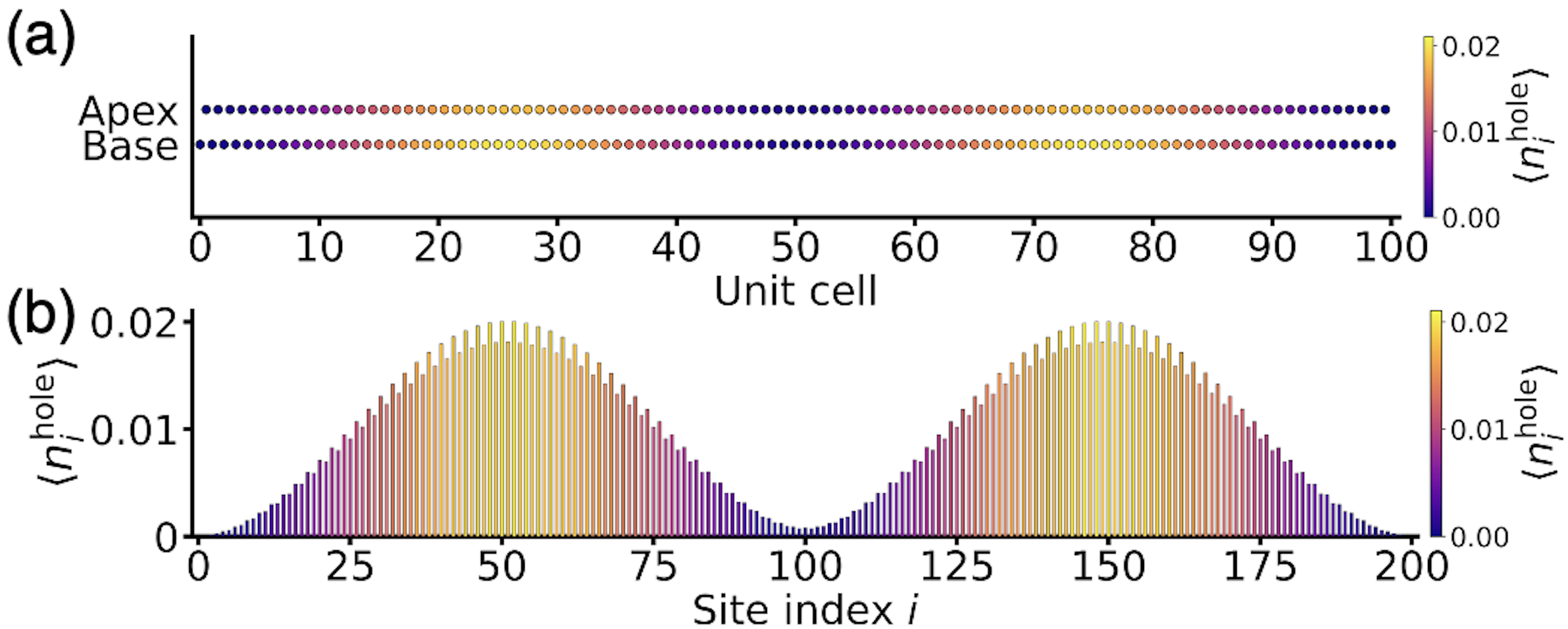}
    \caption{Zero-field two-hole ground-state hole density of the
    sawtooth $U=\infty$ Hubbard chain for $L=100$ triangles
    ($N_{\rm s}=201$), $t_2/t_1=0.1$, $S^z_{\rm tot}=1/2$.
    (a) Sawtooth geometry with color-coded hole density.
    (b) Hole density profile. The central minimum is far deeper than the
    ratio $16/25$ of free spinless fermions,
    Eq.~\eqref{eq:twoorbital}.}
    \label{fig:2holes}
\end{figure}

\subsection{Why one cage wins: a rough estimate}

At fixed $N_\downarrow$ both arrangements have the same Zeeman energy,
so the choice is made by the kinetic confinement energy alone. 
Following the one-hole result, Eq.~\eqref{eq:q_asymptotic},
we model each singlet region as a box with effective
length $\ell+2\delta$, where $\delta$ represents a
boundary extrapolation length. For a long one-hole
cage, $\delta=\coth\kappa_\infty$, which approaches
$L_{\rm loc}$ only when $L_{\rm loc}\gg1$ and remains
equal to one lattice spacing in the compact limit.
Applying a common $\delta$ to shared and separate
two-hole cages is a heuristic approximation.
Relative
to $2E_\infty$, a level $n$ in a box of length $\Lambda$ costs
$\propto n^2/\Lambda^2$. Two separate cages of $N_\downarrow/2$ dimers
each hold one hole in their lowest level,
\begin{equation}
\Delta E_{\rm sep}\propto
\frac{2}{(N_\downarrow/2+2\delta)^2}
=\frac{8}{(N_\downarrow+4\delta)^2}.
\label{eq:two_cages}
\end{equation}
In one shared cage of $N_\downarrow$ dimers, the Pauli principle
places the second hole in the $n=2$ level,
\begin{equation}
\Delta E_{\rm shared}\propto
\frac{1^2+2^2}{(N_\downarrow+2\delta)^2}
=\frac{5}{(N_\downarrow+2\delta)^2}.
\label{eq:shared_cages}
\end{equation}
Both costs scale as $N_\downarrow^{-2}$. For $N_\downarrow\gg\delta$
the shared cage wins by $\propto3/N_\downarrow^2$: halving the box
costs a factor of four for each hole, whereas the second level costs a
factor of four for one hole only. Separate cages have four leaky walls
instead of two, which favors them once the cages become short. From Eqs.~\eqref{eq:two_cages}, \eqref{eq:shared_cages}, separate cages win for
\begin{equation}
N_\downarrow<\frac{4\sqrt5-2\sqrt8}{\sqrt8-\sqrt5}\,\delta
\simeq5.6\,\delta,
\qquad\text{i.e.}\qquad
\ell_{\rm s}<\ell_c\simeq2.8\,\delta
\end{equation}
per cage. 
Within this box model, separate cages can become
competitive when their lengths are comparable to
the boundary extrapolation length. However, this
crossover lies in the short-cage regime, where the
long-wavelength approximation and the treatment
of the holes as free fermions are not controlled.
The estimate therefore provides a qualitative
interpretation, rather than a quantitative prediction
of the separate-cage window. Its location and
existence are determined by the DMRG energy comparison.

\subsection{Summary}


For the three hopping ratios studied, the numerical
results support an extension of the one-hole cage
to a two-hole bound state over a range of fields
below saturation. Both holes occupy a shared singlet
region whose size and relative-separation
distribution remain approximately independent of
system size at fixed positive field.
At $t_2/t_1=0.1$, two separate cages are favored
in a narrow window below saturation; no such
window is found at $t_2/t_1=0.4$ or $0.7$.
The box estimate offers a qualitative explanation
for why separate cages may compete with a shared
cage in the short-cage regime.
An exact solution of the two-hole problem,
which would require going beyond the unique
dimer covering, remains an open question.

\end{document}